\documentclass[twoside,leqno,twocolumn]{article}

\usepackage[letterpaper]{geometry}

\usepackage{ltexpprt}
\usepackage{amsmath}
\usepackage{hyperref}
\usepackage{multirow}

\usepackage{xspace}
\usepackage{niceAlgorithm}
\usepackage{booktabs}
\usepackage{colortbl}
\usepackage{tikz}
\usepackage{pgfplots}
\pgfplotsset{compat=1.16}

\usetikzlibrary{shapes,arrows,pgfplots.statistics,fit}

\newcommand{\atime}{\ensuremath{\tau}\xspace}
\newcommand{\stops}{\ensuremath{\mathcal S}\xspace}
\newcommand{\astop}{\ensuremath{v}\xspace}
\newcommand{\stopa}{\ensuremath{u}\xspace}
\newcommand{\stopb}{\ensuremath{v}\xspace}
\newcommand{\departureTime}{\ensuremath{\atime_\text{dep}}\xspace}
\newcommand{\arrivalTime}{\ensuremath{\atime_\text{arr}}\xspace}
\newcommand{\trips}{\ensuremath{\mathcal{T}}\xspace}
\newcommand{\aTrip}{\ensuremath{T}\xspace}
\newcommand{\aTripA}{\ensuremath{\aTrip_{\mkern-2mu{}a}}\xspace}
\newcommand{\aTripB}{\ensuremath{\aTrip_{\mkern-2mu{}b}}\xspace}
\newcommand{\aTripC}{\ensuremath{\aTrip_{\mkern-2mu{}c}}\xspace}
\newcommand{\aTripD}{\ensuremath{\aTrip_{\mkern-2mu{}d}}\xspace}
\newcommand{\aTripE}{\ensuremath{\aTrip_{\mkern-2mu{}e}}\xspace}
\newcommand{\aTripF}{\ensuremath{\aTrip_{\mkern-2mu{}f}}\xspace}
\newcommand{\aTripG}{\ensuremath{\aTrip_{\mkern-2mu{}g}}\xspace}
\newcommand{\aTripH}{\ensuremath{\aTrip_{\mkern-2mu{}h}}\xspace}

\newcommand{\routes}{\ensuremath{\mathcal R}\xspace}
\newcommand{\aRoute}{\ensuremath{R}\xspace}
\newcommand{\stopEvent}{\ensuremath{\epsilon}\xspace}

\newcommand{\graph}{\ensuremath{G}\xspace}

\newcommand{\vertices}{\ensuremath{\mathcal V}\xspace}
\newcommand{\aVertex}{\ensuremath{v}\xspace}
\newcommand{\bVertex}{\ensuremath{w}\xspace}

\newcommand{\aSource}{\ensuremath{s}\xspace}
\newcommand{\aTarget}{\ensuremath{t}\xspace}
\newcommand{\edges}{\ensuremath{\mathcal E}\xspace}
\newcommand{\edge}{\ensuremath{e}\xspace}
\newcommand{\transferTime}{\ensuremath{\tau_{\textsf{t}}}\xspace}
\newcommand{\transferTimeDiff}{\ensuremath{\Delta_{\textsf{t}}}\xspace}
\newcommand{\finalTransferTime}{\ensuremath{\tau_{\textsf{f}}}\xspace}

\newcommand{\aJourney}{\ensuremath{J}\xspace}

\newcommand{\journeys}{\ensuremath{\mathcal{J}}\xspace}
\newcommand{\restrictedJourneys}{\ensuremath{\journeys_R}\xspace}
\newcommand{\anchorJourney}{\ensuremath{A}\xspace}
\newcommand{\anchorJourneys}{\ensuremath{\journeys_A}\xspace}

\newcommand{\parent}{\text{parent}\xspace}
\newcommand{\children}{\text{children}\xspace}
\newcommand{\candidateKey}{\ensuremath{\atime_\text{c}}\xspace}
\newcommand{\witnessLimit}{\ensuremath{\bar{\atime}}\xspace}
\newcommand{\intermediateWitnessLimit}{\ensuremath{\witnessLimit_1}\xspace}
\newcommand{\finalWitnessLimit}{\ensuremath{\witnessLimit_2}\xspace}
\newcommand{\originStopEvent}{\ensuremath{\stopEvent_\text{o}}\xspace}
\newcommand{\destinationStopEvent}{\ensuremath{\stopEvent_\text{d}}\xspace}
\newcommand{\targetStopEvent}{\ensuremath{\stopEvent_\text{t}}\xspace}
\newcommand{\arrivalWeight}{\ensuremath{w_\text{arr}}\xspace}
\newcommand{\transferWeight}{\ensuremath{w_{\textsf{t}}}\xspace}

\newcommand{\reachedIndex}{\ensuremath{r}\xspace}

\newcommand{\tripBasedEdges}{\ensuremath{\mathcal{E}^{\textsf{t}}}\xspace}
\newcommand{\augmentedEdges}{\ensuremath{\tripBasedEdges_\text{aug}}\xspace}

\newcommand{\aSet}[1]{\{\mkern1mu#1\mkern1mu\}}

\newcommand{\aQueue}{\ensuremath{Q}\xspace}
\DeclareMathOperator{\pred}{pred}

\newcommand{\arrivalSlack}{\ensuremath{\sigma_\text{arr}}\xspace}
\newcommand{\tripSlack}{\ensuremath{\sigma_\text{tr}}\xspace}
\newcommand{\forwardReachedIndex}{\ensuremath{\overrightarrow{\reachedIndex}}\xspace}
\newcommand{\backwardReachedIndex}{\ensuremath{\overleftarrow{\reachedIndex}}\xspace}
\newcommand{\forwardArrivalTime}{\ensuremath{\overrightarrow{\arrivalTime}}\xspace}
\newcommand{\backwardDepartureTime}{\ensuremath{\overleftarrow{\departureTime}}\xspace}

\SetKwFunction{scan}{Scan}
\SetKwFunction{enqueue}{Enqueue}

\newcommand{\gs}{\hphantom{\tiny$\cdot$}}

\newcommand{\legend}[1]{\raisebox{0.07ex}{#1}}

\usepackage{xcolor}
\definecolor{KITgreen}          {rgb}{0,    0.588,0.509}
\definecolor{KITgreen70}        {rgb}{0.3,  0.711,0.656}
\definecolor{KITgreen50}        {rgb}{0.5,  0.794,0.754}
\definecolor{KITgreen30}        {rgb}{0.7,  0.876,0.852}
\definecolor{KITgreen15}        {rgb}{0.85, 0.938,0.926}
\definecolor{KITblue}           {rgb}{0.274,0.392,0.666}
\definecolor{KITblue70}         {rgb}{0.492,0.574,0.766}
\definecolor{KITblue50}         {rgb}{0.637,0.696,0.833}
\definecolor{KITblue30}         {rgb}{0.782,0.817,0.9}
\definecolor{KITblue15}         {rgb}{0.891,0.908,0.95}
\definecolor{KITpalegreen}      {rgb}{0.509,0.745,0.235}
\definecolor{KITpalegreen70}    {rgb}{0.656,0.821,0.464}
\definecolor{KITpalegreen50}    {rgb}{0.754,0.872,0.617}
\definecolor{KITpalegreen30}    {rgb}{0.852,0.923,0.77}
\definecolor{KITpalegreen15}    {rgb}{0.926,0.961,0.885}
\definecolor{KITyellow}         {rgb}{0.98, 0.901,0.078}
\definecolor{KITyellow70}       {rgb}{0.986,0.931,0.354}
\definecolor{KITyellow50}       {rgb}{0.99, 0.95, 0.539}
\definecolor{KITyellow30}       {rgb}{0.994,0.97, 0.723}
\definecolor{KITyellow15}       {rgb}{0.997,0.985,0.861}
\definecolor{KITorange}         {rgb}{0.862,0.627,0.117}
\definecolor{KITorange70}       {rgb}{0.903,0.739,0.382}
\definecolor{KITorange50}       {rgb}{0.931,0.813,0.558}
\definecolor{KITorange30}       {rgb}{0.958,0.888,0.735}
\definecolor{KITorange15}       {rgb}{0.979,0.944,0.867}
\definecolor{KITbrown}          {rgb}{0.627,0.509,0.196}
\definecolor{KITbrown70}        {rgb}{0.739,0.656,0.437}
\definecolor{KITbrown50}        {rgb}{0.813,0.754,0.598}
\definecolor{KITbrown30}        {rgb}{0.888,0.852,0.758}
\definecolor{KITbrown15}        {rgb}{0.944,0.926,0.879}
\definecolor{KITred}            {rgb}{0.627,0.117,0.156}
\definecolor{KITred70}          {rgb}{0.739,0.382,0.409}
\definecolor{KITred50}          {rgb}{0.813,0.558,0.578}
\definecolor{KITred30}          {rgb}{0.888,0.735,0.747}
\definecolor{KITred15}          {rgb}{0.944,0.867,0.873}
\definecolor{KITlilac}          {rgb}{0.627,0,    0.47}
\definecolor{KITlilac70}        {rgb}{0.739,0.3,  0.629}
\definecolor{KITlilac50}        {rgb}{0.813,0.5,  0.735}
\definecolor{KITlilac30}        {rgb}{0.888,0.7,  0.841}
\definecolor{KITlilac15}        {rgb}{0.944,0.85, 0.92}
\definecolor{KITcyanblue}       {rgb}{0.313,0.666,0.901}
\definecolor{KITcyanblue70}     {rgb}{0.519,0.766,0.931}
\definecolor{KITcyanblue50}     {rgb}{0.656,0.833,0.95}
\definecolor{KITcyanblue30}     {rgb}{0.794,0.9,  0.97}
\definecolor{KITcyanblue15}     {rgb}{0.897,0.95, 0.985}
\definecolor{KITseablue}        {rgb}{0.196,0.313,0.549}
\definecolor{KITseablue70}      {rgb}{0.437,0.519,0.684}
\definecolor{KITseablue50}      {rgb}{0.598,0.656,0.774}
\definecolor{KITseablue30}      {rgb}{0.758,0.794,0.864}
\definecolor{KITseablue15}      {rgb}{0.879,0.897,0.932}
\definecolor{KITblack}          {rgb}{0,    0,    0}
\definecolor{KITblack90}        {rgb}{0.1,  0.1,  0.1}
\definecolor{KITblack80}        {rgb}{0.2,  0.2,  0.3}
\definecolor{KITblack75}        {rgb}{0.25, 0.25, 0.25}
\definecolor{KITblack70}        {rgb}{0.3,  0.3,  0.3}
\definecolor{KITblack60}        {rgb}{0.4,  0.4,  0.4}
\definecolor{KITblack50}        {rgb}{0.5,  0.5,  0.5}
\definecolor{KITblack40}        {rgb}{0.6,  0.6,  0.6}
\definecolor{KITblack30}        {rgb}{0.7,  0.7,  0.7}
\definecolor{KITblack25}        {rgb}{0.75, 0.75, 0.75}
\definecolor{KITblack20}        {rgb}{0.8,  0.8,  0.8}
\definecolor{KITblack10}        {rgb}{0.9,  0.9,  0.9}
\definecolor{KITwhite}          {rgb}{1,    1,    1}
\definecolor{LIPICSorange}      {rgb}{0.988,0.78 ,0.07}

\colorlet{nodeColor}{black!80}
\colorlet{edgeColor}{black!50}
\colorlet{axisColor}{black!80}
\colorlet{legendColor}{black!80}

\tikzstyle{vertex}=[circle,line width=.5pt,minimum size=0.1pt]
\tikzstyle{routeArrow}=[->, >=stealth]
\tikzstyle{gridStyle}=[axisColor!20, line width = 0.2pt, dash pattern = on 2pt off 1pt]
\tikzstyle{axisStyle}=[axisColor, line width=0.5pt]
\tikzstyle{markSign} = [mark=o]
\tikzstyle{shortenLines} = [shorten <= 3.4pt,shorten >= 3.4pt]
\tikzstyle{pointsFull} = [mark=*, mark size=1.8pt, line width=1.2pt, only marks]
\tikzstyle{sampleFull} = [mark=*, mark size=1.8pt, line width=1.2pt, shorten <= 3.4pt,shorten >= 3.4pt]
\tikzstyle{legendFull} = [mark=markFull, mark size=1.8pt, line width=1.2pt, shorten <= -1pt,shorten >= -1pt]
\tikzstyle{pointsCirc} = [mark=o, mark size=1.8pt, line width=1.2pt, only marks]
\tikzstyle{sampleCirc} = [mark=o, mark size=1.8pt, line width=1.2pt, shorten <= 3.4pt,shorten >= 3.4pt]
\tikzstyle{legendCirc} = [mark=markCirc, mark size=1.8pt, line width=1.2pt, shorten <= -1pt,shorten >= -1pt]

\begin{document}
\newcommand\relatedversion{}
\renewcommand\relatedversion{\thanks{This research was funded by the DFG under grant number WA 654123-2.}}

\title{\Large Fast Multimodal Journey Planning for Three Criteria\relatedversion}
\author{Moritz Potthoff\,\thanks{\texttt{moritz.potthoff@student.kit.edu}, Karlsruhe Institute of Technology} \and Jonas Sauer\,\thanks{\texttt{jonas.sauer2@kit.edu}, Karlsruhe Institute of Technology}}

\date{}

\maketitle

\begin{abstract} \small\baselineskip=9pt
We study the journey planning problem for multimodal networks consisting of public transit and a non-schedule-based transfer mode (e.g., walking, bicycle, e-scooter).
So far, all efficient algorithms for this problem either restrict usage of the transfer mode or Pareto-optimize only two criteria: arrival time and the number of used public transit trips.
However, we show that both limitations must be lifted in order to obtain high-quality solutions.
In particular, the time spent using the (unrestricted) transfer mode must be optimized as a third criterion.
We present McTB, the first algorithm that optimizes three criteria efficiently by avoiding costly data structures for maintaining Pareto sets.
To enable unlimited transfers, we combine it with a three-criteria extension of the ULTRA~\cite{Bau19b,Sau20b} preprocessing technique.
Furthermore, since full Pareto sets become impractically large for more than two criteria, we adapt an approach by Delling et al.~\cite{Del19} to restrict the Pareto set in a methodical manner.
Extensive experiments on real-world data show that our algorithms are fast enough for interactive queries even on large country-sized networks.
Compared to the state of the art for multicriteria multimodal journey planning, MCR~\cite{Del13}, we achieve a speedup of up to 80.
\end{abstract}

\section{Introduction}
\label{sec:introduction}
With the emergence of new, flexible transport options such as rental bikes and e-scooters, the ability to plan multimodal journeys that combine different transportation modes is becoming more important than ever.
In this work, we consider bimodal transportation networks, consisting of public transit and a \emph{transfer graph} which represents a secondary, non-schedule-based transfer mode (e.g., walking, bicycle, e-scooter).
While many efficient journey planning algorithms have been developed for both road and public transit networks~\cite{Bas16b}, solving the combined multimodal problem is more challenging.
Fully multimodal algorithms such as MCR~\cite{Del13} are slow because they explore the transfer graph with costly Dijkstra searches.
Some public transit algorithms support walking as a transfer mode~\cite{Del12,Del15b,Wit15,Dib13,Dib18}, but only in a limited capacity, for example by requiring the transfer graph to be transitively closed.
Unfortunately, such restrictions lead to much longer travel times for many queries and therefore significantly worsen the solution quality~\cite{Wag17,Pha19}.
A recently proposed approach for lifting these restrictions is ULTRA~\cite{Bau19b,Sau20b}.
It enables any public transit algorithm to support unlimited transfers by precomputing a small set of \emph{transfer shortcuts}.

Like many public transit algorithms, ULTRA supports Pareto optimization of two criteria: arrival time and the number of used public transit trips.
Round-based algorithms such as RAPTOR~\cite{Del12,Del15b} and Trip-Based Routing~(TB)~\cite{Wit15} can handle the second criterion efficiently by avoiding explicit representation of Pareto sets.
By contrast, algorithms for more than two criteria (e.g., McRAPTOR~\cite{Del12,Del15b}) require expensive dynamic data structures to maintain the Pareto sets and are therefore too slow for interactive applications.

For public transit, it is commonly accepted that optimizing only the arrival time is not enough to obtain high-quality solutions.
The number of trips measures the discomfort associated with using public transit:
Since changing vehicles is cumbersome, many passengers will accept a slightly later arrival time to save an additional trip.
With an unlimited transfer mode, we argue that it becomes necessary to also minimize the \emph{transfer time}, i.e., the time that is spent using the transfer mode\footnote{Note that our definition of transfer time excludes the waiting time that may occur when transferring between two vehicles.}.
This is not as important if transfers are restricted, since most journeys will have a low transfer time anyway.
With unlimited transfers, however, journeys that are optimal according to the other two criteria may require an excessive amount of transfer time.
As shown by our experiments in Section~\ref{sec:experiments}, there are often alternatives which are only slightly worse in the other two criteria but require less transfer time.
Such alternatives will be missed when optimizing only two criteria, even with restricted transfers.
Computing satisfactory solutions in a multimodal context therefore requires a third criterion that measures the discomfort associated with the transfer mode.

\subparagraph*{Our Contribution.} To enable efficient three-criteria optimization, we introduce McTB (Section~\ref{sec:full-query-algorithms}), the first algorithm which can optimize a third criterion (e.g., transfer time) without explicit representation of Pareto sets.
Our algorithm is based on TB, which operates on \emph{stop events}, i.e., visits of public transit vehicles at stops.
Since each stop event is associated with a particular arrival time, this removes the need to maintain arrival times explicitly.
McTB builds on this idea by tracking the currently best value for the third criterion at each stop event.
Combined with a round-based exploration of the network, this makes the third criterion the only one whose value must be tracked explicitly.
Consequently, McTB does not require costly dynamic data structures in order to maintain Pareto sets, except for a single set of solutions at the target vertex.

TB requires a preprocessing phase which precomputes transfers between stop events.
As shown in~\cite{Sau20b}, this preprocessing phase can be replaced with a modified version of ULTRA, thereby enabling unlimited transfers.
To apply this to McTB, we develop McULTRA, a three-criteria extension of ULTRA (Section~\ref{sec:shortcut-computation}).
Analogously to ULTRA, McULTRA can be combined with any three-criteria public transit algorithm which normally requires transitively closed transfers.

For three or more criteria, the number of Pareto-optimal journeys becomes impractically large~\cite{Del13,Bas13,Del19}.
Besides causing high query times, this makes it difficult for users to choose between many similar alternatives.
Several techniques have been proposed for filtering the Pareto set~\cite{Del13,Bas13}.
\emph{Restricted Pareto sets}~\cite{Del19} are of particular interest because they can be computed much faster than full Pareto sets, but their definition is independent of the algorithm used to compute them.
Delling et al.~\cite{Del19} introduced BM-RAPTOR, a (Mc)RAPTOR-based algorithm which computes restricted Pareto sets in a network with limited transfers.
In Section~\ref{sec:restricted-query-algorithms}, we show that only minor changes are necessary to make BM-RAPTOR utilize McULTRA shortcuts, thus enabling fast computation of restricted Pareto sets in a multimodal network.
To achieve even faster query times, we re-engineer the pruning scheme of BM-RAPTOR to support (Mc)TB as the underlying query algorithm.

Section~\ref{sec:experiments} presents an experimental evaluation on real-world multimodal networks.
We show that the number of required McULTRA shortcuts remains low when adding transfer time as a third criterion.
Combining McULTRA and McTB yields a speedup of 6--8 compared to MCR, the fastest previously known multimodal algorithm for three criteria.
ULTRA-BM-TB, our new algorithm for restricted Pareto sets, achieves interactive query times even on large networks, yielding a speedup of 30--80 compared to MCR.

\section{Preliminaries}
\label{sec:preliminaries}
Following the notation in~\cite{Bau19b} and~\cite{Sau20b}, a public transit network is a 4-tuple~$(\stops,\trips,\routes,\graph)$ consisting of a set of \emph{stops}~\stops, a set of \emph{trips}~\trips, a set of \emph{routes}~\routes and a directed, weighted~\emph{transfer graph}~$\graph=(\vertices,\edges)$.
Each stop~$\astop\in\stops$ represents a location in the network where passengers can enter~or~exit a vehicle.
A trip~$\aTrip=\langle\stopEvent_0,\dots,\stopEvent_k\rangle\in\trips$ is a sequence of \emph{stop events} performed by a single vehicle.
Each stop event~$\stopEvent_i$ represents a visit of the vehicle at a stop~$\astop(\stopEvent_i)\in\stops$ with \emph{arrival time}~$\arrivalTime(\stopEvent_i)$ and \emph{departure time}~$\departureTime(\stopEvent_i)$.
If a \emph{departure buffer time} must be observed before entering~\aTrip via~$\stopEvent_i$, this can be handled implicitly by reducing~$\departureTime(\stopEvent_i)$ accordingly~\cite{Zuen20}.
The~$i$-th stop event of a trip~\aTrip is denoted by~$\aTrip[i]$ and the number of stop events in the trip by~$|\aTrip|$.
The routes~\routes represent a partition of the trips such that all trips of a route share the same stop sequence and no trip overtakes another along this sequence.
For two trips~$\aTrip_1,\aTrip_2$ of the same route~$\aRoute\in\routes$, we write~$\aTrip_1\preceq\aTrip_2$ if~$\arrivalTime(\aTrip_1[i])\leq\arrivalTime(\aTrip_2[i])$ for every index~$i$ along~\aRoute.
If~$\aTrip_1\preceq\aTrip_2$ and there is at least one index~$j$ with~$\arrivalTime(\aTrip_1[j])<\arrivalTime(\aTrip_2[j])$, we write~$\aTrip_1\prec\aTrip_2$.
Note that trips from different routes are not comparable via~$\preceq$ and~$\prec$.
The unrestricted transfer graph~$\graph=(\vertices,\edges)$ consists of a set of \emph{vertices}~$\vertices$ with~$\stops\subseteq\vertices$, and a set of \emph{edges}~$\edges\subseteq\vertices\times\vertices$.
For each edge~$\edge=(\aVertex,\bVertex)$, the \emph{transfer time}~$\transferTime(\edge)$ is the time required to travel from~\aVertex to~\bVertex along~\edge.

Given a source vertex~$\aSource\in\vertices$ and a target vertex~$\aTarget\in\vertices$, an \aSource-\aTarget-\emph{journey} represents the movement of a passenger from~\aSource to~\aTarget through the network.
It consists of an alternating sequence of \emph{trip legs} (i.e., subsequences of trips) and \emph{transfers} (i.e., paths in the transfer graph, which may be empty).
The transfer between~\aSource and the first trip leg is called the \emph{initial transfer}, whereas the \emph{final transfer} connects the final trip leg to~\aTarget.
The remaining transfers, which connect two trip legs, are called \emph{intermediate transfers}.

\subparagraph*{Problem Statement.}
An ~\aSource-\aTarget-journey~\aJourney is evaluated according to the three criteria arrival time~$\arrivalTime(\aJourney)$ at~\aTarget, number of used trips~$|\aJourney|$ (i.e., the number of trip legs), and transfer time~$\transferTime(\aJourney)$ (i.e., the total time spent traversing the transfer graph).
A journey~\aJourney \emph{weakly dominates} another journey~$\aJourney'$ if~\aJourney is not worse than~$\aJourney'$ according to any of the three criteria.
Moreover,~\aJourney \emph{strongly dominates}~$\aJourney'$ if~\aJourney weakly dominates~$\aJourney'$ and~\aJourney is strictly better than~$\aJourney'$ according to at least one criterion.

Given source and target vertices~$\aSource,\aTarget\in\vertices$ and a departure time~\departureTime, the objective is to compute a (full or restricted) \emph{Pareto set} of \aSource-\aTarget-journeys that depart no earlier than~$\departureTime$.
A (full) Pareto set~\journeys is a set of minimal size such that every valid journey is weakly dominated by a journey in the set.
An \emph{anchor Pareto set}~\anchorJourneys is a Pareto set considering only the two criteria arrival time and number of trips.
Each journey~$\aJourney\in\journeys$ in the full Pareto set has a corresponding \emph{anchor journey}~$\anchorJourney(\aJourney)$, which is the journey in~\anchorJourneys with the highest number of trips not greater than~$|\aJourney|$.
Given a \emph{trip slack}~$\tripSlack\geq{}1$ and an \emph{arrival slack}~$\arrivalSlack\geq{}1$, the \emph{restricted Pareto set}~\cite{Del19} is defined as
\begin{align*}
\restrictedJourneys:=\{ \aJourney\in\journeys\mid|\aJourney|&\leq|\anchorJourney(\aJourney)|\cdot\tripSlack\text{ and }\\ \arrivalTime(\aJourney)-\departureTime&\leq(\arrivalTime(\anchorJourney(\aJourney))-\departureTime)\cdot\arrivalSlack \}.
\end{align*}
The restricted Pareto set contains every journey from the full Pareto set whose arrival time and number of trips do not exceed their respective slack compared to its anchor journey.
In~\cite{Del19}, the restricted Pareto set was defined using absolute slack values (e.g., $\arrivalSlack=60$\,min).
We argue that it is more natural to use relative slacks, which depend on the length of the anchor journey.
Passengers are more willing to take long detours if the journey is already long, whereas a 60-minute detour on a 15-minute journey is not attractive.
All presented algorithms for computing restricted Pareto sets can be easily modified to support absolute slacks instead.

\subparagraph*{Algorithms.}
We give a brief overview of the algorithms we build on.
RAPTOR~\cite{Del12,Del15b} is a round-based algorithm which optimizes arrival time and number of trips in a network with transitively closed transfers.
Round~$i$ finds journeys with exactly~$i$ trips by performing two steps:
First, all routes visiting stops whose earliest arrival time was improved in round~$i-1$ are collected and scanned.
Afterwards, edges in the transfer graph are relaxed.
McRAPTOR extends RAPTOR for an arbitrary number of additional criteria.
For each stop and round, it maintains a \emph{bag} of Pareto-optimal labels instead of a single earliest arrival time.

Bounded McRAPTOR~(BM-RAPTOR)~\cite{Del19} is an extension of RAPTOR which computes restricted Pareto sets.
Three variants of BM-RAPTOR were proposed in~\cite{Del19}; we focus on the fastest variant (Tight-BMRAP).
It performs three phases:
The \emph{forward pruning search} is a two-criteria RAPTOR query which computes the \emph{earliest arrival time}~$\forwardArrivalTime(\astop,i)$ per stop~\astop and round~$i$, and thereby the anchor Pareto set~\anchorJourneys.
Then, a \emph{backward pruning search} is run for each anchor journey, in order of most used trips to fewest.
Collectively, these compute a \emph{latest departure time}~$\backwardDepartureTime(\astop,i)$ for each stop~\astop and round~$i$.
This is the latest departure time at~\astop such that~\aTarget can be reached with~$i$ remaining trips without exceeding the trip or arrival slack.
The backward pruning search for an anchor journey~$\aJourney\in\anchorJourneys$ is a reverse RAPTOR query which starts at~\aTarget with the arrival time~$\departureTime+(\arrivalTime(\aJourney)-\departureTime)\cdot\arrivalSlack$.
Let~$K$ denote the maximum number of trips of any journey in~\anchorJourneys.
The search is run for~$n = \min(|\aJourney'|-1,|\aJourney|\cdot\tripSlack)$ rounds, where~$\aJourney'$ is the journey in~\anchorJourneys with the next higher number of trips, or~$n=K\cdot\tripSlack$ rounds if~$|\aJourney|=K$.
When the backward search finds a departure at a stop~\astop in round~$i$ with departure time~\atime, it is discarded if~$\atime < \forwardArrivalTime(\astop, n-i)$.
Otherwise, it is written to~$\backwardDepartureTime(\astop,K \cdot\tripSlack-(n-i))$.
The latest departure times~$\backwardDepartureTime(\cdot,\cdot)$ are not reinitialized between backward searches.
Finally, the \emph{main search} is a McRAPTOR query from~\aSource to~\aTarget.
If it arrives at a stop~\astop in round~$i$ with arrival time~\atime, the arrival is discarded if~$\atime > \backwardDepartureTime(\astop, K\cdot\tripSlack - i)$.

A faster alternative to RAPTOR is Trip-Based Routing~(TB)~\cite{Wit15}, which requires precomputed transfers between stop events.
Instead of storing arrival times at stops, it maintains a \emph{reached index}~$\reachedIndex(\aTrip)$ for each trip~\aTrip, which is the index of the first stop along~\aTrip that has been reached via some trip~$\aTrip'\preceq\aTrip$.
Like \mbox{RAPTOR}, TB operates in rounds, where each round scans segments of trips which were newly reached in the previous round.
For each stop event along a scanned trip segment, outgoing transfers to other stop events and to the target vertex are relaxed.

ULTRA~\cite{Bau19b,Sau20b} is a preprocessing technique which allows any public transit algorithm that requires a transitively closed transfer graph to handle unlimited transfers instead.
To this end, it precalculates a set of \emph{transfer shortcuts} representing all intermediate transfers required to find a two-criteria Pareto set.
There are two variants of ULTRA:
The \emph{stop-to-stop} variant computes shortcuts between stops.
These are sufficient for most algorithms, including RAPTOR.
The \emph{event-to-event} variant computes shortcuts between stop events, enabling integration with TB~\cite{Sau20b}.
Both variants work by enumerating journeys with at most two trips.
Journeys that consist of two trips connected by an intermediate transfer are called \emph{candidates}, while all other journeys are called \emph{witnesses}.
If a candidate is not dominated by any witness, a shortcut representing its intermediate transfer is added.
An ULTRA query uses Bucket-CH~\cite{Kno07,Gei08,Gei12} to explore initial and final transfers.
Afterwards, a public transit algorithm of choice is run, using the transfer shortcuts as the transfer graph.

\section{McULTRA Shortcut Computation}
\label{sec:shortcut-computation}
To enable unlimited transfers while optimizing a third criterion, we propose McULTRA, an adaptation of ULTRA.
As with original ULTRA, we present two variants:
The stop-to-stop variant can be combined with any algorithm that requires transitively closed transfers, using the same query framework as ULTRA.
The event-to-event variant is intended for combination with McTB, which we present in Section~\ref{sec:full-query-algorithms}.
Since the shortcut computation is mostly identical for both variants, we mainly describe the event-to-event variant and mention differences where appropriate.

The basic outline of the shortcut computation remains unchanged:
For every stop~$\aSource\in\stops$, the algorithm collects all departure times of trips at~\aSource and then runs an \emph{iteration} of McRAPTOR for each departure time in descending order.
The McRAPTOR iterations are restricted to two rounds (each consisting of route scans followed by transfer relaxations), and vertex bags are not cleared between them.
During the final transfer relaxation phase of each iteration, labels representing undominated candidates are extracted and shortcuts are inserted for them.
Adding a third criterion requires the use of McRAPTOR instead of RAPTOR, and therefore vertex bags instead of mere arrival times.
In the following, we adapt the optimizations included in the original ULTRA shortcut computation and describe additional optimizations for the three-criteria setting.
Several optimizations rely on the following simple observation:

\begin{lemma}\label{lemma:stop-event}
    For a fixed number of trips~$i$ and a fixed stop event~$\stopEvent$, the three-criteria Pareto set contains at most one journey with~$i$ trips that uses~\stopEvent as the final stop event.
\end{lemma}
\begin{proof}
    All journeys with~$i$ trips and final stop event~\stopEvent are equivalent according to arrival time and number of trips.
    The best journey according to the third criterion therefore dominates the others.
\end{proof}

\subparagraph{Route Scans.}
When scanning a route, RAPTOR maintains a single \emph{active trip} along the route, which is the earliest trip that can be entered so far.
In McRAPTOR, the active trip is replaced with a \emph{route bag} containing all Pareto-optimal labels, each of which has its own active trip.
No two labels can have the same active trip:
When exiting at any stop along the route, they will share the same final stop event, and by Lemma~\ref{lemma:stop-event}, the label that is better in the third criterion will dominate the other one.

Bags are always required in the second route scanning phase of each McRAPTOR iteration.
For the first route scan, the simpler RAPTOR variant with a single active trip is sufficient as long as the third criterion only affects the transfer graph (which is the case for transfer time).
In this case, because candidates do not have an initial transfer, all candidates labels will have value~$0$ in the third criterion during the first route scan.
Therefore, the candidate with the earliest arrival time dominates all others.
While the same is not true for witnesses, it is not necessary to find all witnesses since failing to find one may only lead to superfluous shortcuts.

\subparagraph{Dijkstra Searches.}
The single-criterion Dijkstra searches of ULTRA are replaced with a standard multicriteria Dijkstra.
Given an \emph{arrival weight}~$\arrivalWeight>0$ and a \emph{transfer weight}~$\transferWeight>0$ (assuming the third criterion is transfer time), a label with arrival time~\arrivalTime and transfer time~\transferTime has a key of~$(\arrivalWeight\cdot\arrivalTime+\transferWeight\cdot\transferTime)/(\arrivalWeight+\transferWeight)$ in the priority queue.
It is also possible to assign a weight of~$0$ to one of the two criteria as long as it is used as a tiebreaker in the case of equality.
This ensures that the order in which labels are settled does not conflict with Pareto dominance, i.e., if label~$\ell_1$ is settled before label~$\ell_2$,~$\ell_2$ may not dominate~$\ell_1$.
The stopping criteria for the Dijkstra searches used by the original ULTRA shortcut computation can be carried over to McULTRA:
The final Dijkstra search of each iteration is stopped once all candidates have been extracted from the queue.
The stopping criterion for the intermediate Dijkstra search is based on the \emph{intermediate witness limit}~\intermediateWitnessLimit.
If~\candidateKey is the key of the last candidate extracted from the queue, the search is stopped once the key of the queue head exceeds~$\candidateKey+\intermediateWitnessLimit$.

\subparagraph{Parent Pointers.}
Two-criteria ULTRA uses per-stop parent pointers to extract shortcuts and to distinguish between candidates and witnesses.
This is no longer sufficient for three criteria since a stop may now have multiple candidate labels which represent different shortcuts.
By Lemma~\ref{lemma:stop-event}, each stop \emph{event} may have at most one candidate.
Thus, McULTRA maintains two pointers~$\parent_1[\stopEvent]$ and~$\parent_2[\stopEvent]$ per stop event~\stopEvent, where~$\parent_k[\stopEvent]$ is the parent stop event for reaching~\stopEvent with~$k$ trips.
Each candidate label includes a pointer to the last stop event where a trip was entered or exited.
To distinguish them from candidates, witness labels set this pointer to~$\bot$. 
When a candidate journey with final stop event~\targetStopEvent is found, the corresponding shortcut can be extracted as~$(\parent_1[\parent_2[\targetStopEvent]], \parent_2[\targetStopEvent])$.

A crucial optimization of ULTRA is that once a shortcut is added to the result, all remaining candidates which represent the same shortcut are turned into witnesses.
To achieve this efficiently, each stop event~\stopEvent maintains a list of child stop events:

\[\children_2[\stopEvent]:=\{\targetStopEvent\mid\parent_2[\targetStopEvent]=\stopEvent\}\]

Once a shortcut~$(\originStopEvent,\destinationStopEvent)$ is inserted, every child~$\targetStopEvent\in\children_2[\destinationStopEvent]$ is turned into a witness by setting~$\parent_2[\targetStopEvent]=\bot$.
Consequently, a label that points to a stop event~$\targetStopEvent\neq\bot$ is only considered a candidate if~$\parent_2[\targetStopEvent]\neq\bot$.
After the final Dijkstra search is stopped, there may be former candidate labels left in the queue which point to a valid final stop event~\targetStopEvent but for which~$\parent_2[\targetStopEvent]$ is~$\bot$.
These labels must be retained as witnesses for later McRAPTOR iterations, but their stop event pointers are set to~$\bot$.
Otherwise, they might be misidentified as candidates if a later iteration sets~$\parent_2[\targetStopEvent]$ to a valid stop event again.

\subparagraph{Final Transfer Pruning.}
The event-to-event variant of original ULTRA uses a weaker pruning rule for candidates than the stop-to-stop variant:
A candidate is discarded if it is strongly dominated by a witness or weakly dominated by a candidate, but not if it is only weakly dominated by a witness.
This difference is carried over to McULTRA.
Preliminary experiments showed that the final Dijkstra search takes much longer in the event-to-event variant than in the stop-to-stop variant, which is not the case for original ULTRA.
The reason for this is that the stopping criterion is applied later due to undominated candidate labels with a very high key.
These labels are only weakly dominated by equivalent labels which were candidates in previous iterations.
To remedy this, we distinguish between \emph{proper candidates}, which are not dominated by any journey, and \emph{improper candidates}, which are weakly dominated by a witness.
We make the stopping criterion of the final Dijkstra search stricter by introducing a \emph{final witness limit}~$\finalWitnessLimit$.
Let~$\candidateKey$ be the key of the last proper candidate extracted from the queue.
The search is stopped once the key of the queue head exceeds~$\candidateKey+\finalWitnessLimit$.
Afterwards, all remaining improper candidates are removed from the queue and shortcuts are inserted for them.

\section{McTB Query Algorithm}
\label{sec:full-query-algorithms}
As a faster alternative to McRAPTOR, we propose McTB, a new three-criteria query algorithm based on TB.
The original preprocessing phase of TB computes event-to-event transfers based on a limited transfer graph.
Since we focus on a multimodal scenario, we use a set~\tripBasedEdges of event-to-event McULTRA shortcuts instead and obtain a multimodal algorithm, ULTRA-McTB.
Pseudocode for the trip enqueuing and scanning procedures of McTB is given in Algorithm~\ref{alg:mc-tb}.
For ease of exposition, we assume that the third criterion optimized by McTB is transfer time.
However, McTB supports any third criterion whose values can be totally ordered.

McTB replaces the reached indices~$\reachedIndex(\cdot)$ used by TB with a transfer time label~$\transferTime(\aTrip,i)$ for each stop event~$\aTrip[i]$ of~$\aTrip$, which represents the minimum transfer time needed to reach~$\aTrip[i]$ (or~$\infty$ if~$\aTrip[i]$ is not reachable).
By Lemma~\ref{lemma:stop-event}, the journey represented by~$\transferTime(\aTrip,i)$ is not dominated by any other journeys ending at~$\aTrip[i]$ found so far.
Two invariants are upheld for~$\transferTime(\aTrip,i)$:
For each~$\aTrip'\succ\aTrip$, $\transferTime(\aTrip',i)\leq\transferTime(\aTrip,i)$ holds since a passenger can reach~$\aTrip'[i]$ by traveling to~$\aTrip[i]$ and waiting.
Similarly, remaining seated in a trip does not increase the transfer time, so~$\transferTime(\aTrip,j)\leq\transferTime(\aTrip,i)$ holds for each~$i<j<|\aTrip|$.

\subparagraph*{Initial Transfer Evaluation.}
Like every ULTRA query, ULTRA-McTB explores initial and final transfers with a Bucket-CH search.
This yields the minimal arrival time~$\arrivalTime(\aSource,\stopa)$ for each stop~\stopa reachable via an initial transfer, and the minimal target transfer time~$\finalTransferTime(\stopb,\aTarget)$ for each stop~\stopb from which~\aTarget is reachable via a final transfer.
Afterwards, the algorithm identifies trips that can be entered via an initial transfer.
This is done as in the original TB query:
For each stop~\astop reachable via an initial transfer and each route visiting~\astop, the earliest reachable trip at~\astop is found via binary search and collected for the first trip scanning phase via the~$\enqueue$ procedure.
Note that two-criteria \mbox{ULTRA-TB} replaces this step with RAPTOR-like route scans, exploiting the fact that the earliest reachable trip of a route never increases along its stop sequence.
However, once a third criterion is introduced, it is no longer sufficient to consider the earliest reachable trip of the route, since entering an earlier trip at a later stop may increase the transfer time.

\subparagraph*{Trip Enqueuing.}
Like the original TB query, the McTB query works in rounds, where round~$n$ scans trip segments that were collected in queue~$\aQueue_n$ during round~$n-1$.
A trip segment~$(\aTrip,j,k,\atime)$ represents the subsequence of~\aTrip from~$\aTrip[j]$ to~$\aTrip[k]$.
Here, the trip segment also stores the transfer time~\atime with which it was reached.
Trip segments are inserted into the queue by the~$\enqueue$ procedure~(lines~\ref{alg:enqueue_begin}--\ref{alg:enqueue_update_end}).
When a trip~$\aTrip$ is entered with transfer time~\atime at index~$i$, the first index where it may be exited is~$j:=i+1$.
If~$\transferTime(\aTrip,j)\leq\atime$ (line~\ref{alg:enqueue_check}),~\aTrip does not improve the transfer time at~$\aTrip[j]$ or any later stop event, so the trip segment is discarded.
Otherwise, the end of the trip segment is set to the last index~$k$ for which~$\transferTime(\aTrip,k)>\atime$ holds~(line~\ref{alg:enqueue_find_end}).
Finally, the trip segment is added to the queue~(line~\ref{alg:enqueue_queue}) and the transfer time labels are updated to~\atime, satisfying the two invariants~(lines~\ref{alg:enqueue_update_begin}--\ref{alg:enqueue_update_end}).

\subparagraph*{Trip Scanning.}
To prevent redundant scans, the trip scanning phase starts with a new pruning step~(lines~\ref{alg:segment_pruning_begin}--\ref{alg:segment_pruning_equal}).
Between enqueuing a trip segment~$x=(\aTrip,j,k,\atime)$ and extracting it from~$\aQueue_n$, another trip segment~$x'=(\aTrip',j',k',\atime')$ with~$\aTrip'\preceq\aTrip$, $j\leq{}j'\leq{}k$ and~$\atime'\leq\atime$ may have been enqueued.
In this case, scanning the portion of~$\aTrip$ between~$j'$ and~$k$ is redundant because it overlaps with~$x'$.
The pruning step identifies redundant portions of~$x$ and adjusts the end index~$k$ accordingly.
If such an~$x'$ exists, either~$\aTrip'\prec\aTrip$ or~$\atime'<\atime$ must hold, since otherwise the check in line~\ref{alg:enqueue_check} would have discarded~$x'$.
If~$\atime'<\atime$, then enqueuing~$x'$ ensures~$\transferTime(\aTrip,j')<\atime$, which is checked in line~\ref{alg:segment_pruning_smaller}.
If $\aTrip'\prec\aTrip$, the trip~$\pred(\aTrip)$ which immediately precedes~\aTrip in the route must have~$\transferTime(\pred(\aTrip),j')\leq\transferTime(\aTrip',j')\leq\atime'\leq\atime$.
This is checked in line~\ref{alg:segment_pruning_equal}.

Following the pruning step, final transfers are evaluated for all stops along the trip segment~$x=(\aTrip,j,k,\atime)$.
Line~\ref{alg:update_journeys} adds new journeys to the set~\journeys of Pareto-optimal journeys at~\aTarget, which is the only dynamic data structure maintained by the algorithm.
Outgoing transfers to other trips are then relaxed in lines~\ref{alg:transfers_begin}--\ref{alg:transfers_enqueue} and the reached trips are enqueued.
Before this is done, target pruning is applied in line~\ref{alg:transfers_pruning}.
Any journey with an intermediate transfer from~$x$ to another trip has an arrival time of at least~$\arrivalTime(\aTrip[j])$, uses at least~$n+1$ trips and has a transfer time of at least~$\atime$.
If a journey with the smallest possible value for all three criteria is dominated by the current Pareto set~\journeys, outgoing transfers do not need to be relaxed.
We make the dominance check more efficient by exploiting the fact that all journeys added to~\journeys by round~$n$ use at most~$n$ trips.
This makes it possible to maintain a \emph{best bag}~$\journeys^*\subseteq\journeys$ of journeys that are Pareto-optimal according to arrival time and transfer time (ignoring the number of trips).
Then, line~\ref{alg:transfers_pruning} only has to check if the journey is dominated by~$\journeys^*$ instead of~$\journeys$.

\begin{algorithm2e}[t]
    \caption{Trip enqueuing and scanning procedures of the McTB query.}\label{alg:mc-tb}
    \SetKwProg{myproc}{Procedure}{}{}
    \myproc{\enqueue{$\aTrip,j,\aQueue,\atime$}}{\label{alg:enqueue_begin}
        \lIf{$\transferTime(\aTrip,j)\leq\atime$}{\Return}\label{alg:enqueue_check}
        $k\leftarrow\max\{i\in\{j,\dots,|\aTrip|-1\}\mid\transferTime(\aTrip,i)>\atime\}$\;\label{alg:enqueue_find_end}
        $\aQueue\leftarrow\aQueue\cup\{(\aTrip,j,k,\atime)\}$\;\label{alg:enqueue_queue}
        \ForEach{$\aTrip'\succeq\aTrip$}{\label{alg:enqueue_update_begin}
            \For{$i$ from $j$ to $|\aTrip|-1$}{
                $\transferTime(\aTrip',i)\leftarrow\min(\atime,\transferTime(\aTrip',i))$\;\label{alg:enqueue_update_end}
            }
        }
    }
    \BlankLine
    \myproc{\scan{$\aQueue_n$}}{\label{alg:scan_begin}
        \ForEach{$(\aTrip\!,j,k,\atime)\in\aQueue_n$}{\label{alg:segment_pruning_begin}
            \For{$i$ from $j+1$ to $k$}{
                \lIf{$\transferTime(\aTrip,i)<\atime$}{$k\leftarrow{}i-1$}\label{alg:segment_pruning_smaller}
                \lElseIf{$\transferTime(\aTrip,i)=\atime$~\And~$\pred(\aTrip)\neq\bot$~\And~$\transferTime(\pred(\aTrip),i)\leq\atime$}{$k\leftarrow{}i-1$}\label{alg:segment_pruning_equal}
            }
        }
        \ForEach{$(\aTrip\!,j,k,\atime)\in\aQueue_n$}{
            \For{$i$ from $j$ to $k$}{
                $\journeys\leftarrow{}$Pareto set of~$\journeys\cup\aSet{(\arrivalTime(\aTrip[i])+\finalTransferTime(\aVertex(\aTrip[i]),\aTarget),n,\atime+\finalTransferTime(\aVertex(\aTrip[i]),\aTarget))}$\;\label{alg:update_journeys}
            }
        }
        \ForEach{$(\aTrip\!,j,k,\atime)\in\aQueue_n$}{\label{alg:transfers_begin}
            \If{$\journeys$~dominates~$(\arrivalTime(\aTrip[j]),n+1,\atime)$}{\Continue}\label{alg:transfers_pruning}
            \For{$i$ from $j$ to $k$}{
                
                \ForEach{\smash{$\edge=(\aTrip[i],\aTrip'[i'])\in\tripBasedEdges$}}{
                    $\enqueue(\aTrip',i'+1,\aQueue_{n+1},\atime+\transferTime(\edge))$\;\label{alg:transfers_enqueue}
                }
            }
        }
    }
    
\end{algorithm2e}

\section{Bounded Query Algorithms}
\label{sec:restricted-query-algorithms}
In this section, we introduce algorithms for computing restricted Pareto sets in a network with unlimited transfers.
Note that simply combining BM-RAPTOR with McULTRA shortcuts is not sufficient.
BM-RAPTOR requires that the forward and backward pruning searches find optimal arrivals or departures at each stop.
However, since ULTRA-RAPTOR explores final transfers with a backward Bucket-CH search from~\aTarget, optimal arrivals via a transfer at stops other than~\aTarget may not be found.
We show that only small modifications are necessary to make ULTRA-RAPTOR support pruning searches.
Furthermore, we introduce a new TB-based algorithm for computing restricted Pareto sets.

\subsection{ULTRA-BM-RAPTOR}
For the pruning searches, we replace ULTRA-RAPTOR with a slightly modified algorithm, which we call pRAPTOR.
Aside from using the three-criteria McULTRA shortcuts for the intermediate transfers, the only difference is the transfer relaxation phase.
Normally, ULTRA-RAPTOR relaxes the outgoing shortcut edges of all stops whose arrival time was improved during the preceding route scanning phase.
pRAPTOR additionally relaxes the edges of stops whose arrival time was improved during the previous transfer relaxation phase.
This allows pRAPTOR to find journeys that use multiple shortcut edges in a row.
However, each additional edge is counted as an additional trip.

Theorem~\ref{th:pruning-raptor} shows that pRAPTOR can be used to perform pruning searches.
Combined with ULTRA-McRAPTOR for the main search, this yields ULTRA-BM-RAPTOR as a multimodal algorithm for computing restricted Pareto sets.

\begin{theorem}\label{th:pruning-raptor}
    Let~$\arrivalTime(\astop,i)$ denote the arrival time at~\astop found by pRAPTOR in round~$i$.
    Consider a journey~\aJourney which is Pareto-optimal for three criteria.
    Given a stop~\astop visited by~\aJourney, let~$\arrivalTime(\aJourney,\astop)$ denote the arrival time of~\aJourney at~\astop.
    For~$1\leq{}i\leq{}|\aJourney|$, let~$\aTrip_i$ be the $i$-th trip of~\aJourney, $\stopa_i$ the stop where~$\aTrip_i$ is entered and~$\stopb_i$ the stop where~$\aTrip_i$ is exited.
    Then~$\arrivalTime(\stopa_i,i-1)\leq\arrivalTime(\aJourney,\stopa_i)$ and~$\arrivalTime(\stopb_i,i)\leq\arrivalTime(\aJourney,\stopb_i)$.
\end{theorem}
\begin{proof}
    First we show that~$\arrivalTime(\stopa_i,i-1)\leq\arrivalTime(\aJourney,\stopa_i)$ implies~$\arrivalTime(\stopb_i,i)\leq\arrivalTime(\aJourney,\stopb_i)$.
    If pRAPTOR arrives at~$\stopa_i$ no later than~\aJourney, it will scan~$\aTrip_i$ or an earlier trip of the same route during the route scanning phase of round~$i$, and thereby reach~$\stopb_i$ with an arrival time of~$\arrivalTime(\aJourney,\stopb_i)$ or better.
    For~$i<|\aJourney|$, we show that this in turn implies~$\arrivalTime(\stopa_{i+1},i)\leq\arrivalTime(\aJourney,\stopa_{i+1})$:
    If the arrival via the route of~$\aTrip_i$ in round~$i$ improves the previous value of~$\arrivalTime(\stopb_i,i)$, then the following transfer phase in round~$i$ will relax the shortcut~$(\stopb_i,\stopa_{i+1})$ and find a suitable arrival at~$\stopa_{i+1}$.
    Otherwise, $\arrivalTime(\stopb_i,j)\leq\arrivalTime(\aJourney,\stopb_i)$ must hold for some~$j<i$.
    Then the transfer phase of round~$j+1\leq{}i$ will relax~$(\stopb_i,\stopa_{i+1})$ and arrive at~$\stopa_{i+1}$ in time.
    Since the base case~$i=1$ follows from the correctness of Bucket-CH and RAPTOR, the claim is proven by induction.
\end{proof}

\subsection{ULTRA-BM-TB}
To achieve even faster query times, we introduce ULTRA-BM-TB, a TB-based algorithm for computing restricted Pareto sets.
ULTRA-BM-TB follows the same query framework as BM-RAPTOR, but uses a variant of ULTRA-TB for the pruning searches and ULTRA-McTB for the main search.
Switching to TB-based algorithms necessitates different data structures for pruning.
The main search of BM-RAPTOR relies on earliest arrival times~$\forwardArrivalTime(\cdot,i)$ and latest departure times~$\backwardDepartureTime(\cdot,i)$ for each round~$i$, which are computed by the pruning searches.
A natural adaptation of these data structures for TB is to replace them with a \emph{forward reached index}~$\forwardReachedIndex(\aTrip,i)$ and a \emph{backward reached index}~$\backwardReachedIndex(\aTrip,i)$ for each trip~$\aTrip$ and round~$i$.
The forward reached index~$\forwardReachedIndex(\aTrip,i)$ is the first stop index along~\aTrip that is reachable from~\aSource with~$i$ trips, while the backward reached index~$\backwardReachedIndex(\aTrip,i)$ is the last stop index along~\aTrip from which~\aTarget is reachable with~$i$ trips and without exceeding the trip or arrival slack.
These reached indices can be used for pruning in the~$\enqueue$ procedure:
When entering a trip~$\aTrip$ at stop event~$\aTrip[k]$ in round~$i$, a backward pruning search running for~$n$ rounds does not enqueue the trip if~$k<\forwardReachedIndex(\aTrip,n-i)$.
Likewise, the main search does not enqueue the trip if~$k>\backwardReachedIndex(\aTrip,K\cdot\tripSlack-i)$.

\subparagraph*{Computing Per-Round Reached Indices.}
The presented pruning scheme requires the pruning searches to output one reached index per stop and round, whereas the original TB query only maintains one reached index per stop.
For the forward pruning search, this can be changed by simply initializing~$\forwardReachedIndex(\astop,i)$ with~$\forwardReachedIndex(\astop,i-1)$ for each stop~\astop at the start of round~$i$.
The backward pruning searches require a different approach because they do not access the rounds of~$\backwardReachedIndex(\cdot,\cdot)$ in order.
Before the first backward search is started, the backward reached indices for all~$K\cdot\tripSlack$ rounds are initialized with~$\infty$.
Whenever a backward reached index~$\backwardReachedIndex(\aTrip,i)$ is set to a value~$k$, this value is propagated to the following rounds by setting~$\backwardReachedIndex(\aTrip,j)$ to~$\min(\backwardReachedIndex(\aTrip,j),k)$ for all~$i<j\leq{}K\cdot\tripSlack$.

\subparagraph*{Shortcut Augmentation.}
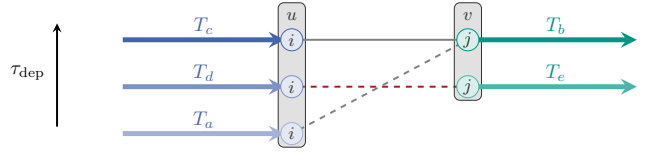
\begin{figure}[t]
    \centering
    \scalebox{0.75}{\begin{tikzpicture}[scale=1.04]
\begin{scope}[yscale=0.8]
    \node (sa) at ( 0.00,  0.00) {};%
    \node (sd) at ( 0.00,  1.00) {};%
    \node (sc) at ( 0.00,  2.00) {};%
    \node (ua) at ( 3.00,  0.00) {};%
    \node (ud) at ( 3.00,  1.00) {};%
    \node (uc) at ( 3.00,  2.00) {};%
    \node (u_label) at (3.00, 2.50) {};%
    \node (vb) at ( 6.00,  2.00) {};%
    \node (ve) at ( 6.00,  1.00) {};%
    \node (v_label) at (6.00, 2.50) {};%
    \node (tb) at ( 9.00,  2.00) {};%
    \node (te) at ( 9.00,  1.00) {};%
    \node (arrow_begin) at (-1.00, 0.00) {};
    \node (arrow_end) at (-1.00, 2.50) {};
    
    \node [fit=(ua)(ud)(uc)(u_label),line width=.5pt, draw=nodeColor!100,fill=nodeColor!15,rounded corners=0.1cm] {};%
    \node [fit=(vb)(ve)(v_label),line width=.5pt,draw=nodeColor!100,fill=nodeColor!15,rounded corners=0.1cm] {};%
    
    \draw [KITblue!50, line width=2.5pt, routeArrow] (sa) -- (ua);
    \draw [KITblue!70, line width=2.5pt, routeArrow] (sd) -- (ud);
    \draw [KITblue, line width=2.5pt, routeArrow] (sc) -- (uc);
    \draw [KITgreen, line width=2.5pt, routeArrow] (vb) -- (tb);
    \draw [KITgreen!70, line width=2.5pt, routeArrow] (ve) -- (te);
    
    \draw [black, line width=1pt, routeArrow] (arrow_begin) -- (arrow_end);
    \node [align=left] at ( -1.50, 1.30) {$\departureTime$};%

    \node [align=left,text=KITblue] at ( 1.50, 0.30) {$\aTrip_a$};%
    \node [align=left,text=KITblue] at ( 1.50, 1.30) {$\aTrip_d$};%
    \node [align=left,text=KITblue] at ( 1.50, 2.30) {$\aTrip_c$};%
    \node [align=left,text=KITgreen] at ( 7.50, 2.30) {$\aTrip_b$};%
    \node [align=left,text=KITgreen] at ( 7.50, 1.30) {$\aTrip_e$};%

    \draw [KITred, line width=1pt, dashed]  (ud) -- (ve);
    \draw [edgeColor, line width=1pt, dashed]  (ua) -- (vb);
    \draw [edgeColor, line width=1pt]  (uc) -- (vb);

    \node (ua_v) at (ua) [vertex,draw=KITblue!50,fill=KITblue!15] {\gs};%
    \node (ud_v) at (ud) [vertex,draw=KITblue!70,fill=KITblue!15] {\gs};%
    \node (uc_v) at (uc) [vertex,draw=KITblue!100,fill=KITblue!15] {\gs};%
    \node at (vb) [vertex,draw=KITgreen!100,fill=KITgreen!15] {\gs};%
    \node at (ve) [vertex,draw=KITgreen!70,fill=KITgreen!15] {\gs};%

    \node at (ua) [text=nodeColor!100] {\small{$i$}};%
    \node at (ud) [text=nodeColor!100] {\small{$i$}};%
    \node at (uc) [text=nodeColor!100] {\small{$i$}};%
    \node at (u_label) [text=nodeColor!100] {\small{$\stopa$}};%
    \node at (vb) [text=nodeColor!100] {\small{$j$}};%
    \node at (ve) [text=nodeColor!100] {\small{$j$}};%
    \node at (v_label) [text=nodeColor!100] {\small{$\stopb$}};%
\end{scope}
\end{tikzpicture}}%
    \caption{%
        Shortcut augmentation in an example network.
        Trips~$\aTrip_a$, $\aTrip_c$ and~$\aTrip_d$ belong to the blue route and visit stop~\stopa at index~$i$.
        Trips~$\aTrip_b$ and~$\aTrip_e$ belong to the green route, which visits stop~\stopb at index~$j$.
        Earlier trips are drawn below later trips of the same route.
        If the shortcut~$(\aTrip_c[i],\aTrip_b[j])$ exists, the shortcut~$(\aTrip_a[i],\aTrip_b[j])$ is added unless a shortcut~$(\aTrip_d[i],\aTrip_e[j])$ exists.
    }%
    \label{fig:shortcutAugmentation}%
\end{figure}

To establish a correct pruning scheme, the pruning searches must fulfill a condition analogous to that stated by Theorem~\ref{th:pruning-raptor}.
For this purpose, we introduce the set of \emph{augmented shortcuts}~\augmentedEdges for a given set~\tripBasedEdges of event-to-event McULTRA shortcuts:
\begin{align*}
\augmentedEdges = \big\{ (\aTripA[i],\aTripB[j]) \mid \exists \aTripC \succeq \aTripA&\colon (\aTripC[i],\aTripB[j]) \in \tripBasedEdges \mkern9mu \land \\\forall \aTripD \succeq \aTripA \,\, \forall \aTripE \prec \aTripB&\colon (\aTripD[i], \aTripE[j]) \notin \tripBasedEdges  \big\}
\end{align*}

Figure~\ref{fig:shortcutAugmentation} illustrates how shortcut augmentation works.
For each shortcut~$(\aTrip_c[i],\aTrip_b[j])\in\tripBasedEdges$ and each trip~$\aTrip_a$ preceding~$\aTrip_c$ in the same route, it ensures that a shortcut exists from~$\aTrip_a[i]$ to either~$\aTrip_b[j]$ or an earlier trip of the same route at~$j$.
Clearly, this can be achieved by adding the shortcut~$(\aTrip_a[i],\aTrip_b[j])$.
However, if a shortcut~$(\aTrip_d[i],\aTrip_e[j])$ exists for~$\aTrip_d\succeq\aTrip_a$ and~$\aTrip_e\prec\aTrip_b$, the shortcut~$(\aTrip_a[i],\aTrip_b[j])$ is superfluous because there must already be a shortcut from~$\aTrip_a[i]$ to~$\aTrip_e[j]$ or an even earlier trip at~$j$.

Lemma~\ref{th:augmentation} and Theorem~\ref{th:pruning-tb} prove that using an ULTRA-TB query with per-round reached indices and augmented McULTRA shortcuts for the pruning searches yields a correct pruning scheme.

\begin{lemma}\label{th:augmentation}
    For any~$(\aTripC[i],\aTripB[j])\in\tripBasedEdges$ and all~$\aTripA\preceq\aTripC$, there exists a~$\aTripF\preceq\aTripB$ with~$(\aTripA[i],\aTripF[j])\in\augmentedEdges$.
\end{lemma}
\begin{proof}
    Consider a trip~$\aTripA\preceq\aTripC$.
    Assume that for every~$\aTripF\preceq\aTripB$, there is no shortcut~$(\aTripA[i],\aTripF[j])$ in~$\augmentedEdges$.
    In particular, this means~$(\aTripA[i],\aTripB[j])\notin\augmentedEdges$.
    Since~$\aTripC\succeq\aTripA$ and~$(\aTripC[i],\aTripB[j])\in\tripBasedEdges$, it follows by definition of~\augmentedEdges that there must be a~$\aTripD\succeq\aTripA$ and~$\aTripE\prec\aTripB$ such that~$(\aTripD[i], \aTripE[j]) \in \tripBasedEdges$.
    Among all such~\aTripE, consider the one with the earliest departure time.
    Then it follows by definition of~\augmentedEdges that~$(\aTripA[i], \aTripE[j])\in\augmentedEdges$ since there is no edge~$(\aTripG[i],\aTripH[j])\in\tripBasedEdges$ with~$\aTripG\succeq\aTripA$ and~$\aTripH\prec\aTripE$.
    This contradicts our assumption.
\end{proof}

\begin{theorem}\label{th:pruning-tb}
    Consider a journey~\aJourney which is Pareto-optimal for three criteria.
    For~$1\leq{}i\leq|\aJourney|$, let~$\aTrip_i$ be the~$i$-th trip of~\aJourney and~$\aTrip_i[k_i]$ the stop event where~\aJourney exits~$\aTrip_i$.
    Then an ULTRA-TB query using augmented McULTRA shortcuts has a reached index~$\reachedIndex(\aTrip_i) \leq k_i$ after round~$i$.
\end{theorem}
\begin{proof}
    By induction over~$i$.
    The base case~$i=1$ follows from the correctness of Bucket-CH and TB.
    Assume the claim is true for~$i-1$.
    Since~$\reachedIndex(\aTrip_{i-1})\leq{}k_{i-1}$ after round~$i-1$, there must be a trip~$\aTrip'_{i-1}\preceq\aTrip_{i-1}$ that was exited at~$\aTrip'_{i-1}[k_{i-1}]$ in a round~$j<i$.
    We know that~$(\aTrip_{i-1}[k_{i-1}],\aTrip_i[\ell_i])\in\tripBasedEdges$ for the stop event~$\aTrip_i[\ell_i]$ where~\aJourney enters~$\aTrip_i$.
    Then~$(\aTrip'_{i-1}[k_{i-1}],\aTrip'_i[\ell_i])\in\augmentedEdges$ for some~$\aTrip'_i\preceq\aTrip_i$ by Lemma~\ref{th:augmentation}.
    Thus,~$\reachedIndex(\aTrip_i)\leq\reachedIndex(\aTrip'_i)\leq\ell_i+1\leq{}k_i$ after round~$i$.
\end{proof}

\subparagraph*{Optimizations.}
As described thus far, the pruning scheme causes unnecessary work during the backward searches:
The initial transfer phase of a backward search collects all routes from which~\aTarget is reachable and scans them to find trip segments to enqueue.
This is efficient for a normal ULTRA-TB query because~\aTarget is typically reachable from most stops.
In a bounded query, however, most stops are not reachable from~\aSource in time to produce an optimal journey when transferring from there to~\aTarget.
For such stops, the~$\enqueue$ operation will be called, but no trip segments will be enqueued because they will be pruned by the forward reached indices.
To avoid unnecessary~$\enqueue$ calls at these stops, we replace the forward reached indices~$\forwardReachedIndex(\cdot,\cdot)$ with earliest arrival times at stops.
As in the RAPTOR-based algorithm, the forward pruning search now computes an earliest arrival time~$\forwardArrivalTime(\astop,i)$ per stop~\astop and round~$i$.
A corollary of Theorem~\ref{th:pruning-tb} is that~$\forwardArrivalTime(\astop,i)$ is never worse than the arrival time of a Pareto-optimal journey at~\astop in round~$i$ via a trip, allowing the backward search to use it for pruning.
When entering a trip~$\aTrip$ at stop event~$\aTrip[k]$ in round~$i$, a backward search running for~$n$ rounds does not enqueue the trip segment if~$\forwardArrivalTime(\astop(\aTrip[k]),n-i)>\departureTime(\aTrip[k])$.
To explore initial transfers, a backward search with~$n$ rounds for a journey~\aJourney first collects all stops~\astop from which~\aTarget is reachable and for which~$\forwardArrivalTime(n)+\transferTime(\astop,\aTarget)-\departureTime\leq(\arrivalTime(\aJourney)-\departureTime)\cdot\arrivalSlack$ holds.
For each such stop~\astop and each route visiting~\astop, the latest reachable trip is found via binary search and the~$\enqueue$ operation is called.

To compute~$\forwardArrivalTime(\cdot,\cdot)$, the forward pruning search adds an extra operation during the trip scanning phase.
When relaxing the outgoing transfers of a stop event~$\aTrip[\ell]$ in round~$i$,~$\forwardArrivalTime(\astop(\aTrip[\ell]),i)$ is now set to the minimum of itself and~$\arrivalTime(\aTrip[\ell])$.
When starting a new round~$i$, the arrival time~$\forwardArrivalTime(\astop,i)$ for each stop~\astop is initialized with~$\forwardArrivalTime(\astop,i-1)$.

A final optimization concerns the main search.
Normally, the transfer times~$\transferTime(\cdot,\cdot)$ maintained by McTB are cleared at the start of each query.
In the context of ULTRA-BM-TB, where most of the search space is pruned, this is often more expensive than the main search itself.
Hence, we mark each transfer time~$\transferTime(\astop,i)$ with a timestamp.
When~$\transferTime(\astop,i)$ is accessed and its timestamp does not match that of the current query, the value is reset to~$\infty$.

\section{Experiments}
\label{sec:experiments}
\begin{table}[t]%
    \caption{%
        Sizes of the public transit networks, including the unrestricted transfer graphs as well as the transitively closed transfer graphs used for comparison with McRAPTOR and BM-RAPTOR.
    }%
    \label{tbl:networks}%
    \begin{tabular}{@{\,}l@{\,}r@{\,}r@{\,}r@{\,\,}}
        \toprule
        & \hspace{20mm}\llap{London} & \hspace{21mm}\llap{Switzerland} & \hspace{20mm}\llap{Germany} \\
        \midrule
        Stops                                 &                    19\,682 &                         25\,125 &                    243\,167 \\
        Routes                                &                     1\,955 &                         13\,786 &                    230\,225 \\
        Trips                                 &                   114\,508 &                        350\,006 &                 2\,381\,394 \\
        Stop events                           &                4\,508\,644 &                     4\,686\,865 &                48\,380\,936 \\[3pt]
        Vertices \hspace{-3mm} &                   181\,642 &                        603\,691 &                 6\,870\,496 \\
        Edges    \hspace{-3mm} &                   575\,364 &                     1\,853\,260 &                21\,367\,044 \\
        Trans. edges                &                3\,212\,206 &                     2\,639\,402 &                22\,571\,280 \\
        \bottomrule
    \end{tabular}
\end{table}

All algorithms were implemented in C\raisebox{0.35ex}{\scriptsize++}17 compiled with GCC 9.3.1 and optimization flag \mbox{-O3}.
Shortcut computations were run on a machine with two 64-core AMD Epyc Rome 7742 CPUs clocked at~2.25\,GHz, with a turbo frequency of~3.4\,GHz,~1024\,GiB of~\mbox{DDR4-3200}~RAM, and~256\,MiB of L3 cache.
All other experiments were conducted on a machine with two 8-core Intel Xeon Skylake SP Gold 6144 CPUs clocked at~3.5\,GHz, with a turbo frequency of~4.2\,GHz,~192\,GiB of~\mbox{DDR4-2666}~RAM, and~24.75\,MiB of L3 cache.
Source code for our algorithms is publicly available\footnote{\url{https://github.com/kit-algo/ULTRA-Transfer-Time}}.
\begin{figure*}[h!]
    \input{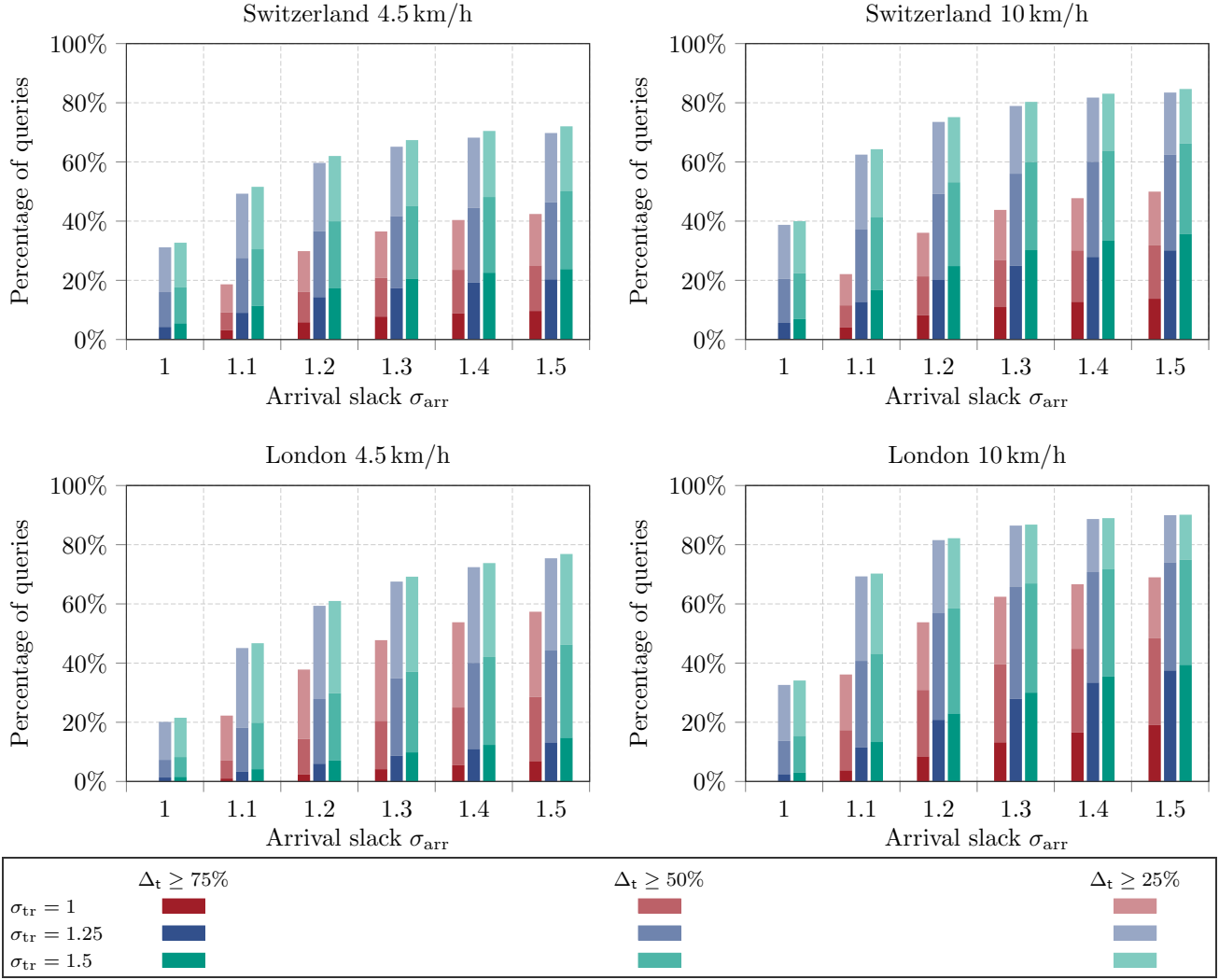}%
    \caption{Comparison of optimal transfer times in the two-criteria Pareto set~$\anchorJourneys$ versus a restricted three-criteria Pareto set~$\restrictedJourneys$ on the Switzerland and London networks for different transfer speeds.
        10\,000 random queries were run for each choice of slack values.
        We then calculated the transfer time savings as~$\transferTimeDiff=(\tau_A-\tau_R)/\tau_A$, where~$\tau_R$ is the lowest transfer time in~$\restrictedJourneys$ and~$\tau_A$ the lowest transfer time in~$\anchorJourneys$.
        Shown is the percentage of queries where~$\transferTimeDiff$ exceeds the specified threshold.}
    \label{fig:transferTimeSavings}
\end{figure*}

Table~\ref{tbl:networks} gives an overview of the networks we use for evaluation.
All three networks were previously used in~\cite{Sau20}.
The London and Switzerland networks were sourced from Transport for London\footnote{\url{https://data.london.gov.uk}} and a publicly available GTFS feed\footnote{\url{https://gtfs.geops.ch/}}, respectively.
The Germany network was provided by Deutsche Bahn.
Unlimited transfer graphs were extracted from OpenStreetMap\footnote{\url{https://download.geofabrik.de/}}.
Unless otherwise noted, travel times were computed by assuming a constant speed of~4.5\,km/h, representing walking.
For comparison with McRAPTOR and BM-RAPTOR, transitively closed transfer graphs were generated using the methods reported in~\cite{Wag17}.
To aid reproducibility, we make the London and Switzerland networks publicly available\footnote{\url{https://i11www.iti.kit.edu/PublicTransitData/ULTRA/}}.
Unfortunately, we cannot provide the Germany network as it is proprietary.

\subparagraph{Impact of Optimizing Transfer Time.}
\begin{table*}[t]
    \center
    \caption{Shortcut computation results for ULTRA and McULTRA.
        Times are formatted as hh:mm:ss.
        For event-to-event McULTRA,~$|\protect\overrightarrow{\augmentedEdges}|$ and~$|\protect\overleftarrow{\augmentedEdges}|$ are the number of augmented forward and backward shortcuts, respectively.}%
    \label{tbl:preprocessing}%
    \begin{tabular*}{\textwidth}{@{\,}l@{\extracolsep{\fill}}r@{\extracolsep{\fill}}r@{\extracolsep{\fill}}r@{\extracolsep{\fill}}r@{\extracolsep{\fill}}r@{\extracolsep{\fill}}r@{\extracolsep{\fill}}r@{\extracolsep{\fill}}r@{\extracolsep{\fill}}r@{\extracolsep{\fill}}r@{\,}}
        \toprule
        \multirow{2}{*}{Network}     & \multirow{2}{*}{Variant} & \multicolumn{2}{c}{ULTRA} & \multicolumn{4}{c}{McULTRA} \\
        \cmidrule(){3-4}  \cmidrule(){5-8}
        &                          & Time      & \#\,Shortcuts & Time & \#\,Shortcuts & $|\overrightarrow{\augmentedEdges}|$ & $|\overleftarrow{\augmentedEdges}|$ \\
        \midrule
        \multirow{2}{*}{London}      & Stop   & 00:03:18 &      150\,738 & 00:14:11 &      136\,755 & -- & -- \\
        & Event & 00:03:57 &  17\,384\,810 & 00:28:35 &  18\,841\,658 & 51\,054\,754 & 58\,136\,150 \\[3pt]
        \multirow{2}{*}{Switzerland} & Stop   & 00:01:51 &      135\,694 & 00:03:48 &      214\,826 & -- & -- \\
        & Event & 00:02:09 &  11\,662\,444 & 00:08:08 &  13\,761\,297 & 44\,417\,661 & 44\,643\,118 \\[3pt]
        \multirow{2}{*}{Germany}     & Stop   & 02:42:59 &   2\,069\,438 &  04:19:50 &   2\,924\,040 & -- & -- \\
        & Event & 02:42:55 & 122\,952\,170 & 10:23:36 & 186\,748\,113 & 557\,609\,345 & 571\,099\,035 \\                           
        \bottomrule
    \end{tabular*}
\end{table*}
\begin{table*}[t]
    \center
    \caption{Performance of full Pareto set algorithms, averaged over~10\,000 random queries.
        RAPTOR query times are divided into phases: initialization, including clearing vertex bags and exploring initial/final transfers (Init.), collecting (Collect) and scanning (Scan) routes, and relaxing intermediate transfers (Relax).
        Also reported are the number of rounds (Rnd.) and the number of found journeys (Jrn.).
        McRAPTOR$^\ast$ only supports stop-to-stop queries with transitive transfers.}%
    \label{tbl:queries-full}%
    \begin{tabular*}{\textwidth}{@{\,}l@{\extracolsep{\fill}}l@{\hspace{-5pt}}c@{\extracolsep{\fill}}r@{\extracolsep{\fill}}r@{\extracolsep{\fill}}r@{\extracolsep{\fill}}r@{\extracolsep{\fill}}r@{\extracolsep{\fill}}r@{\extracolsep{\fill}}r@{\,}}
        \toprule
        \multirow{2}{*}{Network} & \multirow{2}{*}{Algorithm} & \multirow{2}{*}{\shortstack[c]{\vspace{0.04cm}\\Full \vspace{0.15cm} \\ graph \vspace{-0.16cm}}} & \multirow{2}{*}{Rnd.} & \multirow{2}{*}{Jrn.} & \multicolumn{5}{c}{Time [ms]} \\
        \cmidrule(){6-10}
        &                &           & & & Init.  & Collect  & Scan   & Relax  & Total\\
        \midrule
        \multirow{4}{*}{London}              
        & McRAPTOR$^\ast$ & $\circ$    & 12.6 & 21.7 & 16.1 & 2.5 & 92.8  & 157.3 & 268.6\\
        & MCR             & $\bullet$  & 14.6 & 30.5 & 37.5 & 4.5 & 147.6 & 269.5 & 459.2\\
        & ULTRA-McRAPTOR  & $\bullet$  & 14.6 & 30.5 & 27.9 & 3.9 & 144.7 & 53.7  & 230.2\\
        & ULTRA-McTB      & $\bullet$  & 14.8 & 30.5 & --   & --  & --    & --    & 79.2\\[3pt]                                  
        \multirow{4}{*}{Switzerland}              
        & McRAPTOR$^\ast$ & $\circ$    & 24.3 & 20.7 & 20.7 & 8.0  & 120.7 & 132.4  & 281.7\\
        & MCR             & $\bullet$  & 32.7 & 30.5 & 82.8 & 20.3 & 278.7 & 409.9  & 791.7\\
        & ULTRA-McRAPTOR  & $\bullet$  & 32.7 & 30.5 & 56.0 & 18.3 & 272.9 & 66.8   & 414.0\\
        & ULTRA-McTB      & $\bullet$  & 33.5 & 30.5 & --   & --   & --    & --     & 130.2\\[3pt]
        \multirow{4}{*}{Germany} 
        & McRAPTOR$^\ast$  & $\circ$    & 31.0 & 38.3 & 423.8 & 335.4 & 4\,399.7 & 3\,406.9 & 8\,565.8 \\
        & MCR              & $\bullet$  & 36.4 & 57.1 & 929.0 & 454.0 & 7\,192.3 & 15\,554.7 & 24\,131.1 \\
        & ULTRA-McRAPTOR   & $\bullet$  & 36.4 & 57.1 & 716.9 & 430.0 & 7\,366.4 & 1\,647.9 & 10\,160.6 \\
        & ULTRA-McTB       & $\bullet$  & 36.1 & 57.1 & -- & -- & -- & -- & 3\,967.3 \\
        
        \bottomrule
    \end{tabular*}
\end{table*}
To evaluate the impact of optimizing transfer time on the solution quality, we compared minimal transfer times in the two-criteria Pareto set to those in restricted three-criteria Pareto sets.
We chose restricted Pareto sets instead of a full one in order to exclude undesirable journeys with a low transfer time but excessive costs in the other criteria.
Figure~\ref{fig:transferTimeSavings} shows the results on Switzerland and London for different transfer speeds.
For walking as the transfer mode, more than 25\% of the transfer time can be saved for most queries, even with small slack values.
In up to 50\% of all queries, more than half of the transfer time can be saved.
This confirms that adding transfer time as a third criterion often significantly improves the quality of the found journeys.
The savings become even more pronounced as the transfer speed increases.
For 10\,km/h (e.g., slow cycling), the vast majority of queries allow at least a moderate improvement in transfer time.
Often, improvements are possible even without allowing a trip slack.
This underscores that optimizing the transfer time becomes especially crucial as the transfer speed increases.
Fast transfer modes are frequently competitive with public transit in terms of travel time, but not necessarily in terms of comfort.
If the transfer time is not considered, optimal journeys will often avoid trip legs altogether in favor of long transfers.

For walking as the transfer mode, both networks behave similarly.
London has a slightly larger share of queries with savings above 25\%, whereas very large savings are more frequent on the Switzerland network.
The impact of the transfer speed is more drastic for London, where large savings are more common for a transfer speed of 10\,km/h than on the Switzerland network.
This can be explained by different network topologies:
The London network consists of one large metropolitan area.
Here, public transit is common but often features many intermediate stops, so the average travel speed is fairly low across long distances.
If the transfer mode is fast enough, it becomes competitive in terms of travel time.
By contrast, long-distance journeys in Switzerland often involve high-speed trains, where even fast transfer modes are not competitive.

\subparagraph{Shortcut Computation.}
We now evaluate the McULTRA shortcut computation.
Following~\cite{Bau19b}, the transfer graphs were contracted up to an average vertex degree of~20 for Germany and~14 for the other networks.
As mentioned in Section~\ref{sec:shortcut-computation}, the first route scanning phase of each McRAPTOR iteration may use simplified RAPTOR-like route scans, but this may lead to superfluous shortcuts.
On the Switzerland network with both witness limits set to~$\infty$, using McRAPTOR-like route scans yields 4\% fewer shortcuts for the stop-to-stop variant and 1\% for the event-to-event variant.
However, this increases the shortcut computation time by 16\% and 10\%, respectively.
Hence, RAPTOR-like route scans are used for all subsequent experiments.
Similarly, the intermediate witness limit is always set to~$\intermediateWitnessLimit=0$ since this reduces the computation time by~60\% while only producing 2\% more shortcuts.

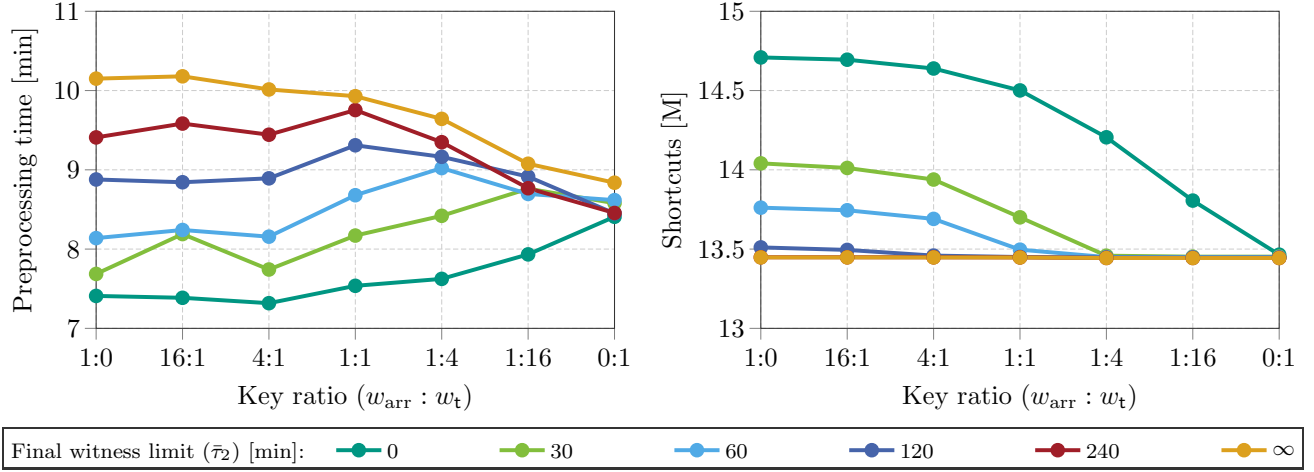
\begin{figure*}
    \newcommand{\plotHeight}{5.8cm}
\newcommand{\plotWidth}{0.49\textwidth}

\begin{tikzpicture}
\arrayrulecolor{legendColor}
\pgfplotsset{
   grid style={KITblack20,line width = 0.2pt,dash pattern = on 2pt off 1pt}
}

\colorlet{plotColor1}{KITgreen}
\colorlet{plotColor2}{KITpalegreen}
\colorlet{plotColor3}{KITcyanblue}
\colorlet{plotColor4}{KITblue}
\colorlet{plotColor5}{KITred}
\colorlet{plotColor6}{KITorange}

\begin{axis}[
   name=mainPlot,
   height=\plotHeight,
   width=\plotWidth,
   xmin=0,
   xmax=6,
   ymin=420,
   ymax=660,
   xlabel={Key ratio ($\arrivalWeight:\transferWeight$)},
   ylabel={Preprocessing time [min]},
   ylabel style = {yshift=-1pt},
   xtick={0, 1, 2, 3, 4, 5, 6},
   xticklabels={$1$:$0$, $16$:$1$, $4$:$1$, $1$:$1$, $1$:$4$, $1$:$16$, $0$:$1$},
   ytick={0, 420, 480, 540, 600, 660},
   yticklabel=\pgfmathparse{(\tick / 60)}${\pgfmathprintnumber{\pgfmathresult}}$\!,
   xtick pos=left,
   ytick pos=left,
   ytick align=outside,
   xtick align=outside,
   grid=both,
   axis line style={legendColor},
   at={(0.0\linewidth,0)}
]

\addplot[color=plotColor1,line width=1.5pt,mark=*] table[x expr=\coordindex,y=Time0,col sep=tab]{fig/keyWitnessLimit.dat};
\addplot[color=plotColor2,line width=1.5pt,mark=*] table[x expr=\coordindex,y=Time1800,col sep=tab]{fig/keyWitnessLimit.dat};
\addplot[color=plotColor3,line width=1.5pt,mark=*] table[x expr=\coordindex,y=Time3600,col sep=tab]{fig/keyWitnessLimit.dat};
\addplot[color=plotColor4,line width=1.5pt,mark=*] table[x expr=\coordindex,y=Time7200,col sep=tab]{fig/keyWitnessLimit.dat};
\addplot[color=plotColor5,line width=1.5pt,mark=*] table[x expr=\coordindex,y=Time14400,col sep=tab]{fig/keyWitnessLimit.dat};
\addplot[color=plotColor6,line width=1.5pt,mark=*] table[x expr=\coordindex,y=TimeInfty,col sep=tab]{fig/keyWitnessLimit.dat};

\end{axis}

\begin{axis}[
   scaled y ticks = false,
   y tick label style={/pgf/number format/fixed},
   height=\plotHeight,
   width=\plotWidth,
   xmin=0,
   xmax=6,
   ymin=13,
   ymax=15,
   xlabel={Key ratio ($\arrivalWeight:\transferWeight$)},
   ylabel={Shortcuts [M]},
   ylabel style = {yshift=-3pt},
   xtick={0, 1, 2, 3, 4, 5, 6},
   xticklabels={$1$:$0$, $16$:$1$, $4$:$1$, $1$:$1$, $1$:$4$, $1$:$16$, $0$:$1$},
   ytick={13, 13.5, 14, 14.5, 15},
   yticklabel=\pgfmathparse{\tick}${\pgfmathprintnumber{\pgfmathresult}}$\!,
   xtick pos=left,
   ytick pos=left,
   ytick align=outside,
   xtick align=outside,
   grid=both,
   axis line style={legendColor},
   at={(0.511\linewidth,0)}
]

\addplot[color=plotColor1,line width=1.5pt,mark=*] table[x expr=\coordindex,y=Shortcuts0,col sep=tab]{fig/keyWitnessLimit.dat};\label{legend:wl:0}
\addplot[color=plotColor2,line width=1.5pt,mark=*] table[x expr=\coordindex,y=Shortcuts1800,col sep=tab]{fig/keyWitnessLimit.dat};\label{legend:wl:30}
\addplot[color=plotColor3,line width=1.5pt,mark=*] table[x expr=\coordindex,y=Shortcuts3600,col sep=tab]{fig/keyWitnessLimit.dat};\label{legend:wl:60}
\addplot[color=plotColor4,line width=1.5pt,mark=*] table[x expr=\coordindex,y=Shortcuts7200,col sep=tab]{fig/keyWitnessLimit.dat};\label{legend:wl:120}
\addplot[color=plotColor5,line width=1.5pt,mark=*] table[x expr=\coordindex,y=Shortcuts14400,col sep=tab]{fig/keyWitnessLimit.dat};\label{legend:wl:240}
\addplot[color=plotColor6,line width=1.5pt,mark=*] table[x expr=\coordindex,y=ShortcutsInfty,col sep=tab]{fig/keyWitnessLimit.dat};\label{legend:wl:max}

\end{axis}

\node[inner sep=0pt,outer sep=0pt,yshift=-3pt] (legend) at (mainPlot.outer south west) {};

\node[inner sep=0pt,outer sep=0pt,anchor=north west] at (legend) {
\footnotesize
\begin{tabular*}{\textwidth}{|@{~}l@{~~~~}r@{\extracolsep{\fill}}r@{\extracolsep{\fill}}r@{\extracolsep{\fill}}r@{\extracolsep{\fill}}r@{\extracolsep{\fill}}r@{~}|}
    \hline
                                         &                               &                                 &                                 &                                   & &  \\[-6pt]
    Final witness limit (\finalWitnessLimit) [min]: & \legend{\ref{legend:wl:0}}\,0 & \legend{\ref{legend:wl:30}}\,30 & \legend{\ref{legend:wl:60}}\,60 & \legend{\ref{legend:wl:120}}\,120 & \legend{\ref{legend:wl:240}}\,240 & \legend{\ref{legend:wl:max}}\,\raisebox{0.25ex}{$\infty$} \\[1pt]
    \hline
\end{tabular*}
};

\arrayrulecolor{black}
\end{tikzpicture}%
    \caption{Impact of key choice and final witness limit~\finalWitnessLimit on the shortcut computation time and the number of shortcuts of event-to-event McULTRA, measured on the Switzerland network.
        For a label with arrival time~\arrivalTime and transfer time~\transferTime, a key ratio of $\arrivalWeight:\transferWeight$ indicates a key value of~$(\arrivalWeight\cdot\arrivalTime+\transferWeight\cdot\transferTime)/(\arrivalWeight+\transferWeight)$.
        A weight of~$0$ indicates that the criterion is only used as a tiebreaker.}
    \label{fig:keyWitnessLimit}
\end{figure*}
\begin{figure*}[h!]
    \newcommand{\plotHeight}{5.8cm}
\newcommand{\plotWidth}{0.49\textwidth}

\begin{tikzpicture}
\arrayrulecolor{legendColor}
\pgfplotsset{
   grid style={KITblack20,line width = 0.2pt,dash pattern = on 2pt off 1pt}
}

\colorlet{plotColor1}{KITred}
\colorlet{plotColor2}{KITgreen}
\colorlet{plotColor3}{KITseablue}
\colorlet{plotColor4}{KITorange}
\colorlet{plotColor5}{KITcyanblue}
\colorlet{plotColor6}{KITpalegreen}
\colorlet{plotColor7}{KITlilac}
\colorlet{plotColor8}{KITblue}

\begin{axis}[
   name=mainPlot,
   scaled y ticks = false,
   x tick label style={/pgf/number format/fixed},
   y tick label style={/pgf/number format/fixed},
   height=\plotHeight,
   width=\plotWidth,
   xmin=1,
   xmax=40,
   ymin=0,
   ymax=20,
   xlabel={Transfer speed [km/h]},
   xlabel style = {yshift=2pt},
   ylabel={Shortcut ratio},
   ylabel style = {yshift=-1pt},
   ytick={0, 5, 10, 15, 20},
   yticklabel style = {yshift=0.8pt},
   yticklabel=\pgfmathparse{\tick}${\pgfmathprintnumber{\pgfmathresult}}$\!,
   xtick pos=left,
   ytick pos=left,
   ytick align=outside,
   xtick align=outside,
   grid=both,
   axis line style={legendColor},
   at={(0,0)}
]

\addplot[color=KITblue,line width=1.5pt,mark=*,mark options={fill=white},mark repeat=2,mark phase=2] table[x=Speed,y=ShortcutsUnlimitedWithIsolated,col sep=tab]{fig/transferSpeedPreprocessingStop.dat};\label{legend:stop:nolimit}
\addplot[color=KITblue,line width=1.5pt,mark=*,mark repeat=2,mark phase=2] table[x=Speed,y=ShortcutsLimitedWithIsolated,col sep=tab]{fig/transferSpeedPreprocessingStop.dat};\label{legend:stop:limit}
\addplot[color=KITgreen,line width=1.5pt,mark=*,mark options={fill=white},mark repeat=2,mark phase=2] table[x=Speed,y=ShortcutsUnlimitedWithIsolated,col sep=tab]{fig/transferSpeedPreprocessingEvent.dat};\label{legend:event:nolimit}
\addplot[color=KITgreen,line width=1.5pt,mark=*,mark repeat=2,mark phase=2] table[x=Speed,y=ShortcutsLimitedWithIsolated,col sep=tab]{fig/transferSpeedPreprocessingEvent.dat};\label{legend:event:limit}

\end{axis}

\node[inner sep=0pt,outer sep=0pt] (legend) at (mainPlot.outer south west) {};

\node[inner sep=0pt,outer sep=0pt,anchor=north west] at (legend) {
\footnotesize
\begin{tabular*}{0.5\textwidth}{|@{~}l@{~}c@{\extracolsep{\fill}}c@{~}|}
    \hline
                    &                                    &                                     \\[-3pt]
                    & Stop-to-stop                       & Event-to-event                      \\[3pt]
    Speed limits    & \legend{\ref{legend:stop:limit}}   & \legend{\ref{legend:event:limit}}   \\[3pt]
    No speed limits & \legend{\ref{legend:stop:nolimit}} & \legend{\ref{legend:event:nolimit}} \\[4pt]
    \hline
\end{tabular*}
};

\begin{axis}[
   height=\plotHeight,
   width=\plotWidth,
   xmin=0,
   xmax=6,
   ymin=0,
   ymax=60,
   xtick={0, 1, 2, 3, 4, 5, 6},
   ytick style={draw=none},
   ytick={10, 20, 30, 40, 50, 60},
   xtick pos=left,
   xticklabels={,,},
   yticklabels={,,},
   xtick align=outside,
   grid=both,
   axis line style={legendColor},
   at={(0.51\linewidth,0)}
]

\addlegendimage{color=KITred50,fill=KITred50,area legend}\label{legend:mr:route}
\addlegendimage{color=KITred70,fill=KITred70,area legend}\label{legend:mr:transfer}
\addlegendimage{color=KITred,fill=KITred,area legend}\label{legend:mr:other}

\addlegendimage{color=KITseablue50,fill=KITseablue50,area legend}\label{legend:ur:route}
\addlegendimage{color=KITseablue70,fill=KITseablue70,area legend}\label{legend:ur:transfer}
\addlegendimage{color=KITseablue,fill=KITseablue,area legend}\label{legend:ur:other}

\addlegendimage{color=KITgreen,fill=KITgreen,area legend}\label{legend:utb}

\end{axis}

\begin{axis}[
   hide axis,
   ybar stacked,
   height=\plotHeight,
   width=\plotWidth,
   xmin=-0.5,
   xmax=5.5,
   ymin=0,
   ymax=3000,
   axis line style={legendColor},
   at={(0.51\linewidth,0)}
]

\addplot[draw=none,fill=KITred50,bar width=5pt] table[x expr=\coordindex-0.2,y=DijkstraRoutesTime,col sep=tab]{fig/transferSpeedQuery.dat};
\addplot[draw=none,fill=KITred70,bar width=5pt] table[x expr=\coordindex-0.2,y=DijkstraTransfersTime,col sep=tab]{fig/transferSpeedQuery.dat};
\addplot[draw=none,fill=KITred,bar width=5pt] table[x expr=\coordindex-0.2,y=DijkstraOtherTime,col sep=tab]{fig/transferSpeedQuery.dat};

\end{axis}

\begin{axis}[
   hide axis,
   ybar stacked,
   height=\plotHeight,
   width=\plotWidth,
   xmin=-0.5,
   xmax=5.5,
   ymin=0,
   ymax=3000,
   axis line style={legendColor},
   at={(0.51\linewidth,0)}
]

\addplot[draw=none,fill=KITseablue50,bar width=5pt] table[x expr=\coordindex,y=ULTRARoutesTime,col sep=tab]{fig/transferSpeedQuery.dat};
\addplot[draw=none,fill=KITseablue70,bar width=5pt] table[x expr=\coordindex,y=ULTRATransfersTime,col sep=tab]{fig/transferSpeedQuery.dat};
\addplot[draw=none,fill=KITseablue,bar width=5pt] table[x expr=\coordindex,y=ULTRAOtherTime,col sep=tab]{fig/transferSpeedQuery.dat};

\end{axis}

\begin{axis}[
	hide axis,
	ybar stacked,
	height=\plotHeight,
	width=\plotWidth,
	xmin=-0.5,
	xmax=5.5,
	ymin=0,
	ymax=3000,
	axis line style={legendColor},
	at={(0.51\linewidth,0)}
	]
	
	\addplot[draw=none,fill=KITgreen,bar width=5pt] table[x expr=\coordindex+0.2,y=TBTime,col sep=tab]{fig/transferSpeedQuery.dat};
	
\end{axis}

\begin{axis}[
   name=mainPlot,
   ybar stacked,
   height=\plotHeight,
   width=\plotWidth,
   xmin=-0.5,
   xmax=5.5,
   ymin=0,
   ymax=3000,
   xlabel={Transfer speed [km/h]},
   ylabel={Query time [s]},
   xlabel style = {yshift=2pt},
   ytick={0, 500, 1000, 1500, 2000, 2500, 3000},
   yticklabel=\pgfmathparse{\tick / 1000}${\pgfmathprintnumber{\pgfmathresult}}$\!,
   xtick={0, 1, 2, 3, 4, 5, 6, 7},
   xticklabels={5, 10, 15, 20, 30, 40},
   xtick pos=left,
   xtick style={draw=none},
   ytick pos=left,
   ytick align=outside,
   xtick align=outside,
   axis line style={legendColor},
   at={(0.51\linewidth,0)},
]
\end{axis}

\node[inner sep=0pt,outer sep=0pt,anchor=north east,xshift=\textwidth] at (legend) {
\footnotesize
\begin{tabular*}{0.49\textwidth}{|@{~}l@{~~~}c@{\extracolsep{\fill}}c@{\extracolsep{\fill}}c@{~}|}
    \hline
                  &                                &                                   &                                \\[-6pt]
                  &             Route              &             Transfer              &             Other              \\[1pt]
    MCR   & \legend{\ref{legend:mr:route}} & \legend{\ref{legend:mr:transfer}} & \legend{\ref{legend:mr:other}} \\[1pt]
    ULTRA-McRAPTOR & \legend{\ref{legend:ur:route}} & \legend{\ref{legend:ur:transfer}} & \legend{\ref{legend:ur:other}} \\[1pt]
    ULTRA-McTB & & & \legend{\ref{legend:utb}} \\[1pt]
    \hline
\end{tabular*}
};

\arrayrulecolor{black}
\end{tikzpicture}%
    \caption{Impact of transfer speed for the Switzerland network.
        \textit{Left:} Ratio of shortcuts compared to a transfer speed of 4.5\,km/h.
        Speed limits in the network were obeyed for the lines with filled circles and ignored for the lines with empty circles.
        \textit{Right:}
        Query performance of~MCR and ULTRA-based algorithms, averaged over 1\,000 random queries.
        Speed limits were obeyed.
        For the RAPTOR-based algorithms, query times are divided into route collecting/scanning, transfer relaxation, and remaining time.
    }
    \label{fig:transferSpeed}
\end{figure*}
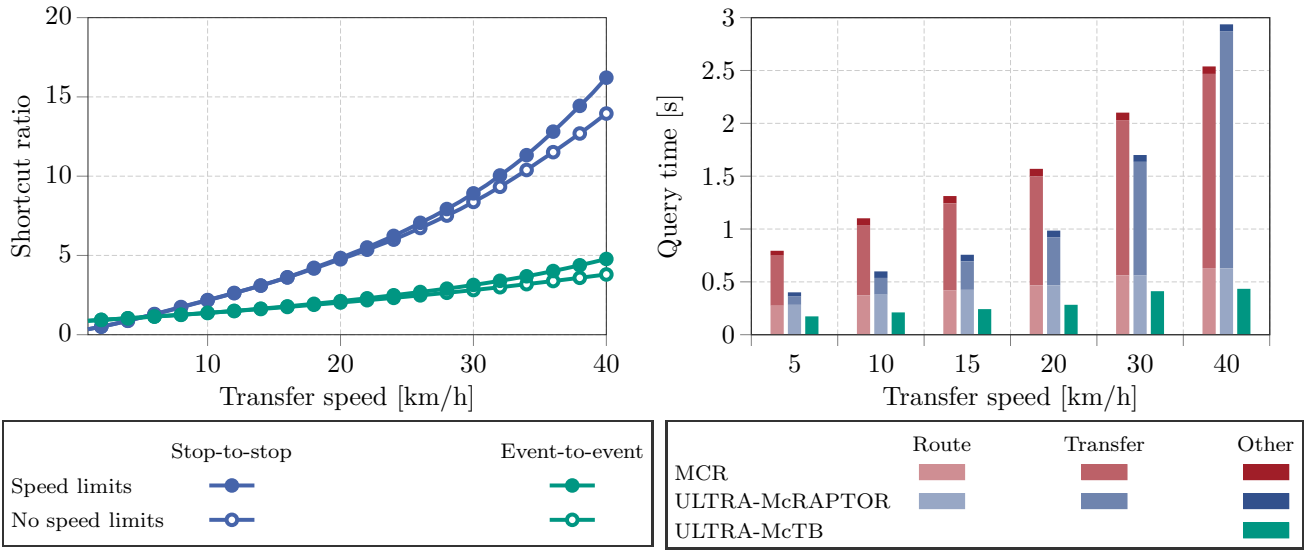
Figure~\ref{fig:keyWitnessLimit} shows the impact of the Dijkstra label key and the final witness limit~\finalWitnessLimit on the computation time and the number of shortcuts in the event-to-event variant.
Especially when the arrival time is weighted heavily, imposing a final witness limit saves considerable computation time, at the cost of producing noticeably more shortcuts.
The latter can be mitigated by weighting the transfer time more heavily, but this increases the computation time for strict final witness limits.
This is explained by the fact that many improper candidates are weakly dominated by witnesses found in previous McRAPTOR iterations, and therefore tend to have higher arrival times than proper candidates.
Increasing the arrival time weight moves proper candidates towards the front of the queue, allowing the stopping criterion to be applied earlier.
To strike a balance between computation time and number of shortcuts, we choose weights~$\arrivalWeight=1$ and $\transferWeight=0$ as well as a final witness limit~$\finalWitnessLimit=60$\,min for all subsequent experiments.

\begin{table*}[h!]
    \center
    \caption{Performance of bounded queries with~$\arrivalSlack=\tripSlack=1.25$, averaged over~10\,000 random queries.
        Query times are divided into forward and backward pruning searches and main search.
        Also reported are the number of rounds performed by the main search (Rnd.) and the number of found journeys (Jrn.).
        Performance of two-criteria algorithms is shown for comparison.
        (BM-)RAPTOR$^\ast$ only supports stop-to-stop queries with transitive transfers.
    }%
    \label{tbl:queries-bounded}%
    \begin{tabular*}{\textwidth}{@{\,}l@{\extracolsep{\fill}}l@{\hspace{-5pt}}c@{\extracolsep{\fill}}r@{\extracolsep{\fill}}r@{\extracolsep{\fill}}r@{\extracolsep{\fill}}r@{\extracolsep{\fill}}r@{\extracolsep{\fill}}r@{\,}}
        \toprule
        \multirow{2}{*}{Network} & \multirow{2}{*}{Algorithm} & \multirow{2}{*}{\shortstack[c]{\vspace{0.04cm}\\Full \vspace{0.15cm} \\ graph \vspace{-0.16cm}}} & \multirow{2}{*}{Rnd.} & \multirow{2}{*}{Jrn.} & \multicolumn{4}{c}{Time [ms]} \\
        \cmidrule(){6-9}
        &                &           & & & Forward  & Backward  & Main   & Total\\
        \midrule
        \multirow{6}{*}{London}
        & RAPTOR$^\ast$    & $\circ$    & 7.8 & 2.5 & -- & -- & -- & 5.2 \\
        & ULTRA-RAPTOR     & $\bullet$  & 8.2 & 4.3 & -- & -- & -- & 6.3 \\
        & ULTRA-TB         & $\bullet$  & 7.1 & 4.3 & -- & -- & -- & 4.3 \\[1pt]
        & BM-RAPTOR$^\ast$ & $\circ$    & 5.3 & 8.0 & 23.2 & 5.2 & 16.4 & 44.7 \\
        & ULTRA-BM-RAPTOR  & $\bullet$  & 4.6 & 12.5 & 10.7 & 3.0 & 7.7 & 21.3 \\
        & ULTRA-BM-TB      & $\bullet$  & 4.6 & 12.5 & 10.0 & 1.8 & 3.8 & 15.5\\[3pt]
        \multirow{6}{*}{Switzerland}
        & RAPTOR$^\ast$    & $\circ$    & 8.3 & 2.5 & -- & -- & --  & 12.3 \\
        & ULTRA-RAPTOR     & $\bullet$  & 8.3 & 4.6 & -- & -- & --  & 12.8 \\
        & ULTRA-TB         & $\bullet$  & 7.5 & 4.6 & -- & -- & --  & 5.0 \\[1pt]
        & BM-RAPTOR$^\ast$ & $\circ$    & 5.7 & 4.8 & 23.9 & 4.4 & 10.9 & 39.3 \\
        & ULTRA-BM-RAPTOR  & $\bullet$  & 4.8 & 9.2 & 21.4 & 3.4 & 8.9 & 33.7 \\
        & ULTRA-BM-TB      & $\bullet$  & 4.8 & 9.2 & 12.1 & 1.4 & 2.5 & 15.9\\[3pt]
        \multirow{6}{*}{Germany}
        & RAPTOR$^\ast$    & $\circ$    & 10.2 & 2.8 & -- & -- & --  & 344.5 \\
        & ULTRA-RAPTOR     & $\bullet$  & 10.4 & 5.3 & -- & -- & --  & 365.2 \\
        & ULTRA-TB         & $\bullet$  & 9.7 & 5.3 & -- & -- & --  & 89.1 \\[1pt]
        & BM-RAPTOR$^\ast$ & $\circ$    & 7.0 & 7.9 & 465.6 & 51.8 & 150.0 & 667.5 \\
        & ULTRA-BM-RAPTOR  & $\bullet$  & 5.9 & 13.8 & 550.4 & 40.8 & 120.8 & 711.9 \\
        & ULTRA-BM-TB      & $\bullet$  & 5.9 & 13.8 & 236.1 & 21.2 & 35.5 & 292.9 \\
        \bottomrule
    \end{tabular*}
\end{table*}
\begin{table*}[h!]
    \center
    \caption{Performance of bounded queries for different slack values on the Switzerland network, averaged over~10\,000 random queries.
        Query times are divided into forward and backward pruning searches and main search.
        Also reported are the number of rounds performed by the main search (Rnd.) and the number of found journeys (Jrn.).
        BM-RAPTOR$^\ast$ only supports stop-to-stop queries with transitive transfers.
    }%
    \label{tbl:queries-bounded-slacks}%
    \begin{tabular*}{\textwidth}{@{\,}l@{\extracolsep{\fill}}l@{\extracolsep{\fill}}l@{\hspace{-5pt}}c@{\extracolsep{\fill}}r@{\extracolsep{\fill}}r@{\extracolsep{\fill}}r@{\extracolsep{\fill}}r@{\extracolsep{\fill}}r@{\extracolsep{\fill}}r@{\,}}
        \toprule
        \multirow{2}{*}{Algorithm} & \multirow{2}{*}{$\tripSlack$} & \multirow{2}{*}{$\arrivalSlack$} & \multirow{2}{*}{\shortstack[c]{\vspace{0.04cm}\\Full \vspace{0.15cm} \\ graph \vspace{-0.16cm}}} & \multirow{2}{*}{Rnd.} & \multirow{2}{*}{Jrn.} & \multicolumn{4}{c}{Time [ms]} \\
        \cmidrule(){7-10}
        &                & &           & & & Forward  & Backward  & Main   & Total\\
        \midrule
        \multirow{8}{*}{BM-RAPTOR$^\ast$} & 1.25 & 1.2  & $\circ$ & 5.7 & 4.5 & 22.9 & 3.4 & 8.3 & 34.5 \\
        & 1.25 & 1.25 & $\circ$ & 5.7 & 4.8 & 23.9 & 4.4 & 10.9 & 39.3 \\
        & 1.25 & 1.3  & $\circ$ & 5.7 & 5.1 & 24.7 & 5.3 & 13.2 & 43.2 \\
        & 1.25 & 1.5  & $\circ$ & 5.7 & 5.8 & 27.8 & 8.5 & 22.4 & 58.7 \\[3pt]
        & 1.5  & 1.2  & $\circ$ & 6.5 & 5.1 & 23.3 & 5.4 & 14.6 & 43.2 \\
        & 1.5  & 1.25 & $\circ$ & 6.5 & 5.4 & 24.1 & 6.9 & 19.2 & 50.2 \\
        & 1.5  & 1.3  & $\circ$ & 6.5 & 5.9  & 24.5 & 8.1 & 23.1 & 55.7 \\
        & 1.5  & 1.5  & $\circ$ & 6.5 & 6.9  & 26.3 & 12.9 & 39.0 & 78.2 \\[8pt]
        \multirow{8}{*}{ULTRA-BM-RAPTOR}   & 1.25 & 1.2  & $\bullet$ & 4.8 & 8.7 & 21.0 & 2.9 & 7.5 & 31.4 \\
        & 1.25 & 1.25 & $\bullet$ & 4.8 & 9.2 & 21.4 & 3.4 & 8.9 & 33.7 \\
        & 1.25 & 1.3  & $\bullet$ & 4.8 & 9.6 & 22.2 & 4.0 & 10.2 & 36.3 \\
        & 1.25 & 1.5  & $\bullet$ & 4.8 & 10.6 & 24.5 & 6.5 & 15.9 & 46.8 \\[3pt]
        & 1.5  & 1.2  & $\bullet$ & 5.6 & 9.3 & 20.7 & 3.4 & 8.9 & 33.0 \\
        & 1.5  & 1.25 & $\bullet$ & 5.6 & 10.0 & 21.4 & 4.2 & 10.7 & 36.2 \\
        & 1.5  & 1.3  & $\bullet$ & 5.6 & 10.5 & 22.3 & 4.8 & 12.4 & 39.5 \\
        & 1.5  & 1.5  & $\bullet$ & 5.6 & 11.9 & 24.3 & 7.7 & 20.2 & 52.3 \\[8pt]
        \multirow{8}{*}{ULTRA-BM-TB}       & 1.25 & 1.2  & $\bullet$ & 4.8 & 8.7 & 11.3 & 1.1 & 1.7 & 14.1 \\
        & 1.25 & 1.25 & $\bullet$ & 4.8 & 9.2 & 12.1 & 1.4 & 2.5 & 15.9 \\
        & 1.25 & 1.3  & $\bullet$ & 4.8 & 9.6 & 12.9 & 1.7 & 3.3 & 17.9 \\
        & 1.25 & 1.5  & $\bullet$ & 4.8 & 10.6 & 16.5 & 3.3 & 7.4 & 26.7 \\[3pt]
        & 1.5  & 1.2  & $\bullet$ & 5.7 & 9.3 & 11.7 & 1.3 & 2.4 & 15.5 \\
        & 1.5  & 1.25 & $\bullet$ & 5.7 & 10.0 & 12.5 & 1.7 & 3.5 & 17.8 \\
        & 1.5  & 1.3  & $\bullet$ & 5.7 & 10.5 & 13.3 & 2.1 & 4.8 & 20.3 \\
        & 1.5  & 1.5  & $\bullet$ & 5.7 & 11.9 & 15.7 & 4.0 & 10.6 & 30.4 \\
        \bottomrule
    \end{tabular*}
\end{table*}

\newpage
Table~\ref{tbl:preprocessing} shows overall results for the shortcut computation on all networks.
Compared to two-criteria \mbox{ULTRA}, the number of shortcuts increases by less than a factor of~2.
The shortcut computation takes 2--3 times longer in the stop-to-stop variant and 4--5 times longer in the event-to-event variant, with slightly higher values for the London network.
By contrast, three-criteria MCR has a slowdown of about~20 compared to its two-criteria variant, MR-$\infty$, indicating that the shortcut computation scales much better for the additional criterion than MCR.

\subparagraph*{Full Pareto Queries.}
We evaluated two query algorithms for computing full Pareto sets: ULTRA-McRAPTOR (combining McRAPTOR with stop-to-stop McULTRA shortcuts) and ULTRA-McTB (combining McTB with event-to-event McULTRA shortcuts).
Query times are reported in Table~\ref{tbl:queries-full}.
Compared to MCR, the fastest previously known algorithm, ULTRA-McRAPTOR is about twice as fast and ULTRA-McTB achieves a speedup of 6--8.
As in the two-criteria scenario, ULTRA-McRAPTOR achieves its speedup mostly by significantly reducing the transfer relaxation time, yielding similar query times to \mbox{McRAPTOR} on a transitively closed transfer graph.
The slight slowdown compared to transitive \mbox{McRAPTOR} is explained by the significantly higher number of found journeys, indicating that unlimited transfers are especially crucial in a multicriteria setting.

\subparagraph*{Impact of Transfer Speed.}
Figure~\ref{fig:transferSpeed} shows how McULTRA is affected by the speed of the transfer mode.
For two-criteria ULTRA, Baum et al.~\cite{Bau19} observed that the number of shortcuts declines once the transfer speed becomes competitive with public transit.
With arrival time and number of trips as criteria, there is no reason to use public transit unless it saves travel time, so fewer shortcuts are produced for higher transfer speeds.
This is no longer the case when adding transfer time as a criterion, since making a public transit detour can save transfer time.
Consequently, the number of shortcuts increases for high transfer speeds.
This effect is much stronger for the stop-to-stop variant than for the event-to-event variant, indicating that most of the additional shortcuts are only required at a few specific times during the day.
All query algorithms become slower for higher transfer speeds as the search space increases.
ULTRA-McRAPTOR is practical for transfer speeds up to 20\,km/h (which is faster than the average travel speed via bicycle or e-scooter), but loses its advantage over MCR for transfer speeds above~30\,km/h.
By contrast, ULTRA-McTB maintains a speedup of 5--6 even for transfer speeds up to 40\,km/h, due to the slower increase in the number of shortcuts.

\subparagraph*{Bounded Queries.}
We conclude by evaluating our bounded query algorithms.
As shown in Table~\ref{tbl:preprocessing}, shortcut augmentation increases the number of shortcuts by a factor of~3.
The augmentation takes less than a minute even for Germany and is therefore negligible compared to the shortcut computation time.
Query times for the bounded algorithms are shown in Table~\ref{tbl:queries-bounded}.
Based on the results from Figure~\ref{fig:transferTimeSavings}, we chose slack values~$\arrivalSlack=\tripSlack=1.25$ as a good tradeoff between solution quality and query speed.
The bounded algorithms are a factor of 2--4 slower than their two-criteria counterparts.
This is consistent with the slowdown observed for \mbox{BM-RAPTOR} compared to RAPTOR on limited transfer graphs.
Most of the running time is taken up by the forward pruning search, indicating that the pruning scheme is highly effective.
Due to weaker target pruning and the algorithmic changes discussed in Section~\ref{sec:restricted-query-algorithms}, the forward pruning search is slower than a two-criteria query, particularly for ULTRA-BM-TB.
Nevertheless, ULTRA-BM-TB is up to 2.5 times faster than ULTRA-BM-RAPTOR.
Compared to the ULTRA-based algorithms for full Pareto sets, the bounded variants achieve a speedup of around an order of magnitude.
Compared to MCR, which was previously the fastest algorithm for this problem setting, ULTRA-BM-TB achieves a speedup of 30--80.

Table~\ref{tbl:queries-bounded-slacks} shows the performance of the bounded query algorithms for different slack values on the Switzerland network.
The forward pruning search is not significantly impacted by the trip slack, and only moderately by the arrival slack, due to the weaker target pruning.
The backward pruning searches and main search take significantly longer for higher slack values, and start to dominate the overall running time for~$\tripSlack=1.5$ and~$\arrivalSlack=1.5$.
However, even for high slack values, the bounded algorithms remain much faster than their unbounded counterparts.

\section{Conclusion}
\label{sec:conclusion}
We showed that in order to improve the solution quality in a multimodal network with unlimited transfers, it is necessary to optimize transfer time as a third criterion.
To achieve this, we developed McTB, a fast three-criteria algorithm that avoids costly dynamic data structures.
To enable unlimited transfers, we proposed McULTRA, a three-criteria extension of ULTRA.
Our shortcut computation algorithm runs in reasonable time and produces less than twice as many shortcuts as two-criteria ULTRA.
The combination of McULTRA and McTB achieves a speedup of 6--8 over the state of the art and remains practical even for fast transfer modes.
Finally, we developed RAPTOR- and TB-based algorithms for computing restricted Pareto sets in a multimodal network by using ULTRA shortcuts.
The latter is up to 80 times faster than MCR, the fastest previously known algorithm for three-criteria multimodal queries.

Future work could involve adapting ULTRA for additional criteria, such as fare or occupancy rate.
It is unclear whether an efficient TB-based algorithm can be designed for more than three criteria.
However, the results reported in~\cite{Del19} suggest that BM-RAPTOR scales well for additional criteria since the pruning searches take up the majority of the running time.
Therefore, combining TB-based pruning searches with a McRAPTOR-based main search may be a promising approach for designing an efficient bounded query algorithm.

\bibliographystyle{plainurl}
\bibliography{arxiv}
\end{document}